\documentclass[11pt, a4paper]{article}

\usepackage{pdfpages}
\usepackage{amsmath}
\usepackage{amsthm}
\usepackage{natbib}
\usepackage{graphicx}
\usepackage{authblk}
\usepackage[left= 2.5cm, right = 2.5cm, top = 3cm, bottom = 2.5cm]{geometry}
\usepackage{color}
\usepackage{url}
\usepackage{booktabs}
\usepackage{multirow}
\usepackage{siunitx}
\usepackage{orcidlink}
\usepackage{pdfpages}
\usepackage{xr-hyper}
\usepackage{hyperref}
\definecolor{ik}{RGB}{11,114,195}
\definecolor{az}{RGB}{167,70,183}

\newtheoremstyle{example} %Name
{3pt} %Space above
{3pt} %Space below
{} %Body font
{0\parindent} %Indent amount 1
{\bf}
{:} %Punctuation after theorem head
{.5em} %Space after theorem head 2
{} %Theorem head spec (
\newtheoremstyle{theorem}% name
{3pt} %Space above
{3pt} %Space below
{\em} %Body font
{0\parindent} %Indent amount 1
{\bf}
{:} %Punctuation after theorem head
{.5em} %Space after theorem head 2
{} %Theorem head spec (
\theoremstyle{example}
\newtheorem{example}{Example}
\theoremstyle{theorem}
\newtheorem{theorem}{Theorem}

\def\E{\mathop{\rm E}}
\def\var{\mathop{\rm var}}

\let\pkg\emph
\let\code\texttt

\title{Extended-support beta regression for $[0, 1]$ responses}

\author[1]{Ioannis Kosmidis~\orcidlink{0000-0003-1556-0302}\thanks{ioannis.kosmidis@warwick.ac.uk}}
\author[2]{Achim Zeileis~\orcidlink{0000-0003-0918-3766}\thanks{Achim.Zeileis@R-project.org}}

\affil[1]{Department of Statistics, University of Warwick \authorcr Coventry, CV4 7AL, UK}
\affil[2]{Department of Statistics, University of Innsbruck \authorcr 6020 Innsbruck, Austria}

\begin{document}

\maketitle

\begin{abstract}
    We introduce the XBX regression model, a continuous mixture of
  extended-support beta regressions for modelling bounded responses
  with boundary observations. The core building block of XBX
  regression is the extended-support beta distribution, a censored
  version of a four-parameter beta distribution with the same
  exceedance on the left and right of $(0, 1)$. Hence, XBX regression
  is a direct extension of beta regression. We prove that beta
  regression and heteroscedastic normal regression with censoring at
  both $0$ and $1$ --- also known as the heteroscedastic two-limit
  tobit model in the econometrics literature --- are special cases of
  extended-support beta regression, depending on whether a single
  extra parameter is zero or infinity, respectively. To overcome
  identifiability issues due to the similarity of the beta and normal
  distributions for certain parameter values, we shrink towards beta
  regression by letting the extra parameter have an exponential
  distribution with unknown mean. A Gauss-Laguerre quadrature
  approximation results in efficient likelihood-based estimation and
  inference procedures, which the \pkg{betareg} R package implements
  since version 3.2.0. We use XBX regression to analyze investment
  decisions in a behavioural economics experiment about the occurrence
  and extent of loss aversion. In contrast to standard approaches, we
  capture both the probability of rational behaviour and the mean of
  loss aversion. Extensive comparisons with alternative approaches
  illustrate the effectiveness of the new model.
  \\
  \noindent {Keywords: {\em Boundary observations; approximate likelihood; tobit regression; loss aversion}}
\end{abstract}

\section{Introduction}

\subsection{Bounded responses and beta regression}
Bounded variables, such as rates, proportions, and fractional indices,
are abundant in data-analytic applications. As with other variable
types, linking them to available covariate information is a popular
data-analytic task, for both inference on covariate effects and
prediction. A range of statistical models has been developed in this
direction, ranging from standard least squares approaches on
transformed responses and other partially-specified models to fully
parametric approaches; see \citet{kieschnick+mccullough:2003} for an
early review and discussion on available approaches.

From the available approaches, beta regression
\citep{ferrari+cribari-neto:2004, smithson+verkuilen:2006} has emerged
as the dominant likelihood-based approach for modelling bounded
responses in terms of covariates.
It has been successfully applied in diverse fields, including biology, ecology, economics and social
sciences, finance, manufacturing, genetics, and engineering; see, for example, \citet{allenbrand+sherwood:2023}, \citet{leng+li+eser:2021}, \citet{namin+xu+zhou:2020}, and \citet{wang+malthouse+krishnamurthi:2015} for a few noteworthy applications, and \citet{douma+weedon:2019} and \citet{geissinger+khoo+richmond:2022} for recent articles making a case for its use in natural sciences and associated applied fields.
The reason for its popularity is that beta regression respects the bounded
range of the responses, and allows their mean and dispersion to be
expressed in terms of covariates through separate regression
structures. Such a construct offers high flexibility in the shape of
the response distribution, capturing both the heteroscedasticity and
the asymmetries observed in bounded responses.

Specifically, consider a random variable $Y$ with a beta distribution
with density
\begin{equation}
\label{eq:beta}
f_{\rm (B)}(y \mid \mu,\phi) =
\frac{I(0 < y < 1)}{B(\mu\phi, (1 -
  \mu)\phi)}y^{\mu\phi-1}(1-y)^{(1 - \mu)\phi-1}
\end{equation}
where $\mu \in (0, 1)$, $\phi > 0$, $B(p, q)$ $(q, p > 0)$ is the beta
function \citep[see,][\S 6.2.1]{abramowitz+stegun:1964}, and $I(A)$ is
the indicator function taking value $1$ if $A$ is true, and zero,
otherwise.  In this parameterization, $\E(Y \mid \mu, \phi) = \mu$ and
$\var(Y \mid \mu, \phi) = \mu(1-\mu)/(1+\phi)$ so that $\phi$ is a
precision parameter. The beta regression model assumes that the
observed responses $y_1, \ldots, y_n$ are realizations of beta random
variables $Y_1, \ldots, Y_n$ that are independent conditionally on
explanatory vectors $x_1, \ldots, x_n$ and $z_1, \ldots, z_n$, where
$x_i$ and $z_i$ are $p$- and $q$-dimensional vectors, respectively,
possibly sharing some or all components. As is apparent, we are assuming here
that the observed responses take value between zero and one or have
been transformed to do so as $y_i = (y_i' - a)/(b - a)$, where $y_i'$
are the original responses and $a$ and $b$ are the minimum and maximum
value they can take. The mean and precision of $Y_i$ are then linked
to the explanatory variables according to
\begin{equation}
\label{eq:link}
  g_1(\mu_i) = \eta_i =   x_i^\top \beta \quad \text{and} \quad
  g_2(\phi_i) = \zeta_i =   z_i^\top \gamma\, ,
\end{equation}
where $\beta = (\beta_1, \ldots, \beta_p)^\top$ and
$\gamma = (\gamma_1, \ldots, \gamma_q)^\top$ are the vectors of
regression parameters associated with the means and the precision of
the beta variables, respectively.

The functions $g_1(\cdot)$ and $g_2(\cdot)$ are monotone functions,
typically with the property of mapping the range of $\mu_i \in (0, 1)$
and $\phi_i \in (0, \infty)$, respectively, to the real line.  Usual
choices for $g_1(\cdot)$ are the logit, probit, and generally any
inverse of a cumulative distribution function, and for $g_2(\cdot)$
the log function.

The parameter vectors $\beta$ and $\gamma$ are typically estimated by
maximum likelihood, and standard inferential procedures, such as
likelihood ratio, Wald, and score tests, apply on the grounds of
central limit theorems for the derivatives of the log-likelihood under
assumptions on model adequacy. \citet{gruen+kosmidis+zeileis:2012}
provide more details and describe a range of modelling strategies and
methods based on the beta regression model.

\subsection{Boundary observations}
\label{sec:boundary}

One important limitation of the beta regression model as defined
by~(\ref{eq:beta}) and (\ref{eq:link}) is that likelihood-based
estimation and inference is not possible when at least one of the
observed responses is either zero or one. Specifically, depending on
the parameter values, the density (\ref{eq:beta}) at those
observations is either zero or infinity. As a result and regardless of
the sample size, the log-likelihood of the beta regression model is
either zero or infinite or undetermined, and maximum likelihood
estimation and the associated inferential procedures cannot be
used. This is the reason that popular software for beta regression
modelling, such as the \pkg{betareg} R package \citep{betareg,
  gruen+kosmidis+zeileis:2012} prior to version 3.2-0, returned an
error when there is at least one response value that is zero or one.

One popular and straightforward approach to deal with boundary
observations is by \citet{smithson+verkuilen:2006}, who propose to
transform the response to $(y_i + u) / (1 + 2 u)$, with
$u = 1 / \{2 (n - 1)\}$, and then fit the beta regression model
in~(\ref{eq:beta}) and~(\ref{eq:link}) with standard software. The
transformed responses, then lie in
$(0.5/n, 1 - 0.5/n) \subset (0, 1)$, and, as a result, there are no
issues when fitting the beta regression model.  The proposal in
\citet{smithson+verkuilen:2006} has the advantage of preserving the
parsimony and the convenience in estimation and inference that is
provided by beta regression model in~(\ref{eq:beta})
and~(\ref{eq:link}), but the value of $u$ is ad hoc and, as
Example~\ref{ex:reading-skills} demonstrates, can have a marked effect
on inference.

\begin{example}
  \label{ex:reading-skills}
  \citet[Example~3]{smithson+verkuilen:2006} models the contribution
  of nonverbal IQ and dyslexic versus nondyslexic status to the scores
  of 44 children on a test of reading accuracy. The scores were
  rescaled to lie in the unit interval, and 13 out of 44 children
  achieved a score of one.  We transform the responses using
  $u = 1 / 86$, and fit the beta regression model with
  $g_1(\mu_i) = \log \{ \mu_i / (1 - \mu_i) \}$,
  $g_2(\phi_i) = \log(\phi_i)$ in~(\ref{eq:link}), where $\eta_i$
  involves an intercept and the main and interaction effects of IQ and dyslexia status,
  and $\zeta_i$ involves only the main effects of IQ and dyslexia status.\footnote{The numerical
    figures in \citet[Table~5]{smithson+verkuilen:2006} can be
    reproduced by transforming only the boundary responses, and
    leaving the non-boundary ones unaltered. The same is true with
    \citet[Table 6.2]{smithson+merkle:2013}.} We focus on the evidence
  against the omission of the interaction term.
  \begin{figure}[t!]
    \begin{center}
      \includegraphics[width = \textwidth]{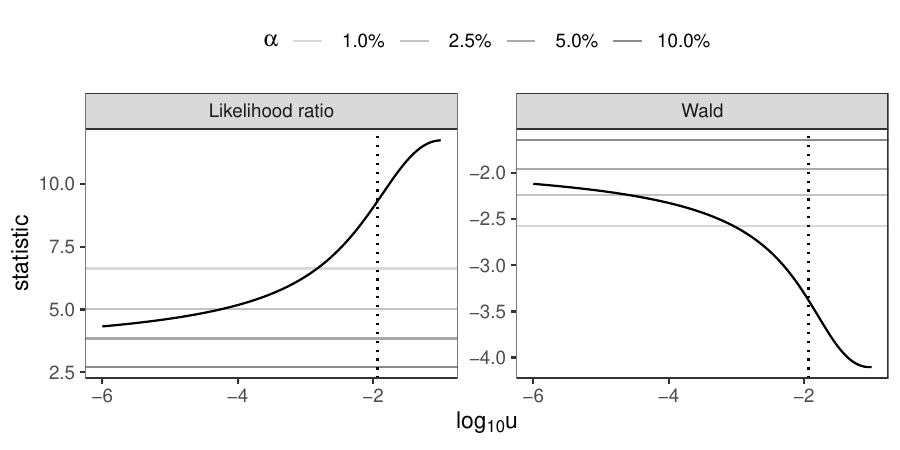}
    \end{center}
    \caption{Likelihood ratio and Wald statistics for the inclusion of
      an interaction term between dyslexia status and IQ in the
      setting of Example~\ref{ex:reading-skills} versus $u$
      ($\log_{10}$ scale). The vertical lines are the quantiles of the
      $\chi^2_1$ (likelihood ratio test) and normal distribution
      (two-sided Wald test) at levels $1.0\%$, $2.5\%$, $5.0\%$,
      $10.0\%$. The vertical dashed line is at $u = 1 / 86$ as per the
      proposal in \citet{smithson+verkuilen:2006}.}
    \label{fig:sensitivity}
  \end{figure}
  Figure~\ref{fig:sensitivity} shows the value of the likelihood ratio
  and Wald statistics for $u$ ranging from $10^{-6}$ to $10^{-1}$. As
  is apparent, both tests are sensitive to the choice of $u$ with the
  evidence for the interaction term ranging from weak to
  strong. Notably, the closer the transformed responses get to the
  observed responses the weaker the evidence against the omission of
  an interaction term.
\end{example}

Another widely-used approach to deal with boundary observations is to
construct two- or three-part hurdle models. Such models are formally
mixtures of point masses at zero and/or one with a continuous
distribution for the non-boundary observations.  \citet{hoff:2007},
\citet{cook+kieschnick+mccullough:2008}, \citet{ospina+ferrari:2010},
\citet{calabrese:2012}, \citet{ospina+ferrari:2012} define such a
model by assuming a beta regression model for the non-boundary
observations and use a regression structure for the hurdle part(s) of the
model. The \pkg{zoib} \citep{liu+kong:2015} and \pkg{gamlss}
\citep{rigby+stasinopoulos:2005} R packages can fit one- and two-part
hurdle models\footnote{Following the terminology set in the seminal
  work of \citet{ospina+ferrari:2010}, the two- and three-part hurdle
  models are typically referred to as zero-or-one and zero-and-one
  inflated beta regression, respectively. We chose the term hurdle
  models instead, because the distributions in
  \citet{ospina+ferrari:2010} can be accurately described as zero
  and/or one being hurdles, which when passed a $(0, 1)$ realization
  of a beta random variable is generated. In contrast, an inflated
  distribution is customarily understood as mixing the process that
  generates boundary values, with a distribution whose support already
  includes the boundary values \citep[see, for
  example,][]{min+agresti:2002}.}.  Hurdle models are a natural choice
when it can be assumed that the boundary observations are generated by
a distinct process from the one that generates the non-boundary
observations. The separate regression structures for the hurdle part(s)
of the model allow for linking the mixing probabilities
directly with explanatory variables, offering both flexibility and
convenience in estimation. If the beta regression part and the hurdle
part(s) do not share any parameters then the likelihood is
separable and can be maximized by separately estimating a binomial or
multinomial regression model (depending on whether only one or both
boundaries get point masses), and a beta regression model for the
non-boundary observations.

Nevertheless, due to the many regression specifications involved in
two- and three-part hurdle models, such models tend to be less
parsimonious than the proposal in \citet{smithson+verkuilen:2006}, which
involves only mean and dispersion regression structures. Furthermore,
hurdle models may be unnecessarily complex and hard to interpret, when
there is merit in assuming that there is a single process generating
the responses. For example, in Example~\ref{ex:reading-skills}, it is
more natural to assume that there is a latent ability scale in terms
of reading accuracy, which is linked to an observable test score with
a known minimum and maximum value. A score of one represents that the
child's ability is sufficiently large to successfully complete all
tasks in the given test. A hurdle model, in contrast, estimates separate
processes for the perfect scores and all other scores (including those
that are close to perfect).

Double censoring is another well-known approach for dealing with
bounded-domain responses with boundary observations, particularly when
it is natural to assume that there is a single process generating
the response values. For double censoring, the observed bounded-domain
response is assumed to be a realization of
$Y = \max(\min(Y^*, 1), 0)$, where $Y^*$ is a random variable with
density function $f(\cdot \mid \theta)$. The marginal expectation of $Y$
is then a linear combination of the censoring probabilities and the
expectation of the truncated distribution of $Y^*$ in $(0,
1)$. Specifically,
\begin{equation}
  \label{eq:censored_mean}
  \E(Y) = P(Y^* \ge 1 \mid \theta) + P(0 < Y^* < 1 \mid
  \theta)\mathop{\rm E}(Y^* \mid 0 < Y^* < 1, \theta) \, .
\end{equation}
The likelihood about $\theta$ based on a sample $y_1, \ldots, y_n$ is
then
\begin{equation}
  \label{eq:likelihood}
  \prod_{i: y_i = 0} P(Y^*_i \le 0 \mid \theta) \prod_{i: y_i = 1} P(Y^*_i \ge 1 \mid \theta) \prod_{i: y_i \in (0, 1)} f(y_i \mid \theta) \, ,
\end{equation}
where $\prod_{A} a$ is $1$ if $A$ is the empty set. The standard
approach towards defining a censored regression model is to put a
regression structure on the components of $\theta$. The most
widely-used censored regression model is the censored normal one,
also known as the two-limit tobit. This model assumes that
$Y^*_1, \ldots, Y^*_n$ are independent conditionally on covariate
vectors $x_1, \ldots, x_n$, with $Y^*_i$ distributed according to a
normal distribution with mean $x_i^\top \beta^*$ and variance
$\sigma_i^2 = \sigma^2$ \citep[see, for example][Section~6.7, for
definition and discussion]{rosett+nelson:1975, maddala:1983}. The
heteroscedastic version of the two-limit tobit assumes that the
variances are linked to explanatory variables as
$g(\sigma_i) = z_i^\top\gamma^*$, where $g(\cdot)$ is a monotone link
function that maps $(0, \infty)$ to the real line \citep[see, for
example,][]{messner+mayr+zeileis:2016}. \citet{dixon+sonka:1982}
compare least squares and maximum likelihood estimation of the
two-limit tobit in terms of estimates in an empirical application, and
\citet{hoff:2007} compares the two-limit tobit model to the
semi-parametric approach in \citet{papke+wooldridge:1996} and the
inflated beta distribution in \citet{cook+kieschnick+mccullough:2008}
for data envelopment analysis. \citet{dixon+sonka:1982},
\citet{hoff:2007}, and \citet*[\S 19.3]{greene:2011} touch on the
interpretation of the parameters of the two-limit tobit. The
disadvantage of the two-limit tobit is that it is less flexible than
the beta distribution and cannot accommodate features that often
occur in bounded-domain variables such as high skewness or kurtosis.

\subsection{Our contribution}

We develop the extended-support beta mixture (XBX) regression model,
which accommodates boundary values of $0$, $1$ or both $0$ and $1$ for
the responses, while preserving the parsimony and flexibility of beta
regression. In contrast to the heteroscedastic two-limit tobit model,
XBX regression has the flexibility of beta regression in the shape of
the response distribution, allowing both for the heteroscedasticity
and the asymmetries that may be observed in bounded responses.

The main building block for XBX regression is the extended-support
beta (XB) distribution, which is constructed by first extending the
support of the beta distribution to $(-u, 1 + u)$, $u > 0$, and then
censoring at both $0$ and $1$. Thus, the XB distribution can be
understood as a four-parameter beta distribution, restricted to
symmetric exceedances and subsequently censored to obtain point masses
on the boundaries of the unit interval. We show that the XB
distribution converges to a beta distribution as $u \to 0$, and to a
normal distribution censored in $[0, 1]$, as $u \to \infty$. Hence,
regression models based on the XB distribution bridge beta regression
and censored normal models like the heteroscedastic two-limit tobit,
at the expense of a single unknown constant $u$. Because the beta
distribution can assume similar shapes to the normal distribution in
its support for certain parameter values, it is natural to shrink XB
towards beta. We achieve that through the XBX distribution, which is a
continuous mixture of XB distributions with $u$ having an exponential
distribution with unknown mean $\nu$. Then, as the mean $\nu$ of the
exponential mixing distribution decreases, XBX gets closer to beta.

The likelihood of the XBX regression model can be conveniently
approximated using a Gauss-Laguerre quadrature, giving rise to
efficient estimation and inference procedures based on maximum
approximate likelihood.

Finally, we present a case study highlighting the benefits of XBX
regression. We revisit the analysis of an experiment from behavioural
economics about the occurrence and extent of loss aversion
\citep{glaetzleruetzler+sutter+zeileis:2015}, which was originally
analyzed using linear regression for the mean loss aversion
only. Here, we demonstrate that the XBX model leads to qualitatively
the same results for the effects of mean loss aversion, simultaneously
allowing for economically relevant insights about the probability of
rational behaviour.

The \pkg{betareg} R package, since version 3.2-0, provides methods for
fitting and drawing inferences from XBX regression models.

Section~\ref{sec:xb} presents the XB distribution and shows that the
beta and the censored normal distribution are its limiting
cases. Section~\ref{sec:xbx} introduces the XBX regression model, and
discusses parameter estimation and inference using approximate
likelihood. Section~\ref{sec:lossaversion} illustrates the benefits of
XBX regression through the analysis of a loss aversion experiment from
behavioral economics.
% The XBX model leads to qualitatively the same
% results for the effects about mean loss aversion as the original
% analysis in \cite{glaetzleruetzler+sutter+zeileis:2015} but
% additionally allows for economically relevant insights about the
% probability of rational behavior.
Section~\ref{sec:xbx-vs-htobit}
presents a comprehensive simulation experiment comparing XBX
regression to the heteroscedastic two-limit tobit model based on the
continuous ranked probability score. We conclude with remarks
and directions for further research in Section~\ref{sec:conculding}.

\section{Extended-support beta distribution}
\label{sec:xb}

Rescaling the response variable as proposed by \citet{smithson+verkuilen:2006},
as briefly summarized in Section~\ref{sec:boundary}, can also be usefully thought of as modelling the
original responses using a four-parameter beta distribution with the same, ad-hoc
exceedance on the left and the right of $(0, 1)$. The four-parameter
beta distribution, denoted as ${\rm B4}(\mu, \phi, u_1, u_2)$, has
density
\begin{equation}
  \label{eq:4parameterbeta}
  f_{\rm (B4)}(y \mid \mu, \phi, u_1, u_2) = f_{ \rm (B)}\left(\frac{y - u_1}{u_2 - u_1} \mid \mu, \phi \right) \frac{1}{u_2 - u_1} \, ,
  % \frac{I(u_1 < y <
  %   u_2)}{B\left(\mu\phi, (1 - \mu)\phi\right)} \left(\frac{y - u_1}{u_2 - u_1} \right)^{\mu\phi -1}
  % \left(\frac{u_2 - y}{u_2 - u_1} \right)^{(1- \mu)\phi -1} \frac{1}{u_2 - u_1} \,,
\end{equation}
where ${u}_1 < {u}_2$ are the boundary parameters. 

\citet[Chapter
22]{johnson+kotz+balakrishnan:1995} describe the characteristics of
the four-parameter beta distribution and provide a coherent literature
review of the challenges and identifiability issues that are involved
in the estimation of the boundary parameters. Apart from the
normalizing constant $1 / (u_2 - u_1)$, the
\citet{smithson+verkuilen:2006} proposal is equivalent to setting
$-u_1 = u_2 - 1 = u = 1/\{2(n - 1)\}$. This makes
density~(\ref{eq:4parameterbeta}) have support on
$(-0.5 / (n - 1), 1 + 0.5 / (n - 1))$, and hence results in positive
likelihood contributions even from zero and/or one
responses. Figure~\ref{fig:xbeta} shows the density of the
standard beta ${\rm B}(\mu, \phi)$ and that of ${\rm B4}(\mu, \phi, -u, 1 + u)$ with
$u \in \{0.01, 0.1, 0.5\}$ and various settings for $\mu$ and
$\phi$. In addition to the inferential issues that are illustrated in
Example~\ref{ex:reading-skills}, the resulting model cannot be
considered as a generative model for $y$ because it has wider support
than $[0, 1]$.

\begin{figure}[t!]
  \begin{center}
    \includegraphics[width = \textwidth]{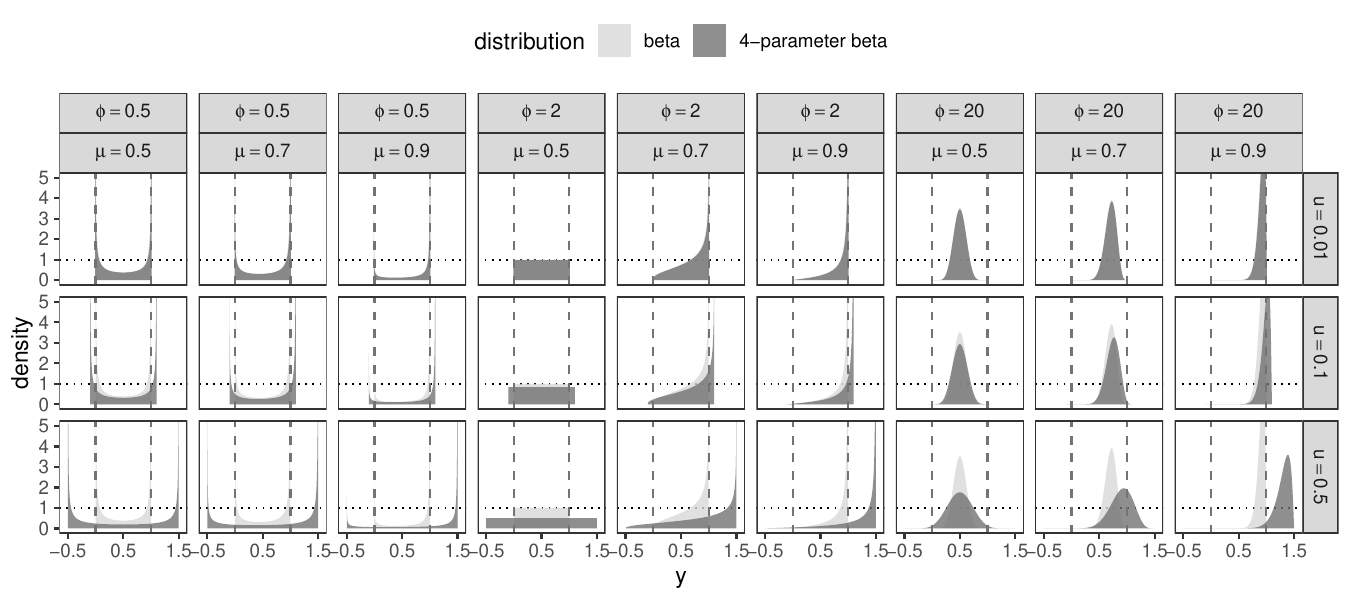}
  \end{center}
  \caption{Densities of the beta ${\rm B}(\mu, \phi)$ and the four-parameter beta
    ${\rm B4}(\mu, \phi, -u, 1 + u)$ distribution for all combinations of $\mu \in \{0.5, 0.7, 0.9\}$, $\phi \in \{0.5, 2, 20\}$, and $u \in \{0.01, 0.1, 0.5\}$. The horizontal dotted line is at 1, and the vertical dashed lines at $0$ and $1$.}
  \label{fig:xbeta}
\end{figure}

Theorem~\ref{th:beta4} below shows that the transformation suggested
in \citet{smithson+verkuilen:2006} has an attractive property;
${\rm B4}(\mu, \phi, -u, 1 + u)$ is the beta distribution for $u = 0$,
and converges to the normal distribution as $u \to \infty$, or,
equivalently, as its support approaches the real line.

\begin{theorem}
  \label{th:beta4}
  Suppose that $Y$ is distributed according to the four-parameter beta
  distribution with density~(\ref{eq:4parameterbeta}) and
  $u = -u_1 = u_2 - 1$. Then,
  \begin{enumerate}
  \item \label{eq:zero} If $u = 0$, $Y$ has the beta density
    $f_{\rm (B)}(y \mid \mu, \phi)$ in~(\ref{eq:beta}).
  \item \label{eq:infinity} If
    $\mu_* = \mathop{\rm E}(Y \mid \mu, \phi, u) \in \Re$ and
    $\sigma_*^2 = \var(Y \mid \mu, \phi, u) \in (0, \infty)$ are fixed,
    then, as $u \to \infty$, $Y$ converges in distribution to
    $\mu_* + \sigma_* Z$, where $Z$ is a standard normal random
    variable.
  \end{enumerate}
\end{theorem}
\begin{proof}
  Part~\ref{eq:zero} is direct by setting $u_1 = 0$ and $u_2 = 1$
  in~(\ref{eq:4parameterbeta}). For part~\ref{eq:infinity}, let
  $\alpha = \mu \phi$ and $\beta = (1- \mu)\phi$. For fixed $\mu_*$
  and $\sigma_*$, $\alpha$ and $\beta$ grow quadratically with $u$, so
  $\alpha \to \infty$ and $\beta \to \infty$ as $u \to \infty$. From
  the properties of the four-parameter beta distribution
  $\mu_* = (1 + 2u) \mu - u$ and $\sigma_*^2 = (1 + 2u)^2 \sigma^2$,
  where $\sigma^2 = \mu (1 - \mu) / (1 + \phi)$. Hence,
  $(y + u) / (1 + 2u) = \mu + \sigma t$, where
  $t = (y - \mu_*) / \sigma_*$. Then, writing the
  density~(\ref{eq:4parameterbeta}) in terms of $t$ and multiplying it
  by $d y / d t$ gives that the density of
  $T = (Y - \mu_*) / \sigma_*$ is
  \begin{equation}
    \label{eq:dens_t}
    h(t \mid \mu, \phi, u) = \frac{\sigma (\mu + \sigma t)^{\alpha - 1}(1 - \mu - \sigma t)^{\beta - 1}}{B(\alpha, \beta)} I\left(\frac{-u - \mu_*}{\sigma_*} < t < \frac{1 + u - \mu_*}{\sigma_*}\right) \, .
  \end{equation}
  For $u \to \infty$, the indicator function on the right hand side
  of~(\ref{eq:dens_t}) has value $1$ for any finite~$t$. Furthermore,
  \citet[Lemma~A.1]{moscovich+etal:2016} shows that the first factor
  on the right hand side of~(\ref{eq:dens_t}) converges to
  $e^{-t^2/2} / \sqrt{2\pi}$ for any $t$ as $\alpha \to \infty$ and
  $\beta \to \infty$. By Scheff\'e’s theorem, $T$ converges in
  distribution to $Z$, and by the continuous mapping theorem $Y$
  converges in distribution to $\mu_* + \sigma_* Z$.
\end{proof}

One way of defining a distribution with $[0, 1]$ support is to use the
censored version of the four-parameter beta distribution. Suppose that
$Y^*$ has a ${\rm B4}(\mu, \phi, -u, 1 + u)$ distribution. We call the
distribution of $Y = \max(\min(Y^*, 1), 0)$ the extended-support beta
distribution, and denote it as ${\rm XB}(\mu, \phi, u)$. The density
function of the XB distribution is
\begin{equation}
  \label{eq:xb}
  f_{\rm (XB)}(y \mid \mu, \phi, u) = \left\{
    \begin{array}{ll}
    \displaystyle F_{\rm (B)}\left(\frac{u}{1 + 2u} \mid \mu, \phi\right)\,, & \text{if } y = 0 \\[1em]
    \displaystyle f_{\rm (B)}\left(\frac{y + u}{1 + 2u} \mid \mu, \phi\right)\frac{1}{1 + 2u}\,, & \text{if } y \in (0, 1) \\[1em]
    \displaystyle 1 - F_{\rm (B)}\left(\frac{1 + u}{1 + 2u} \mid \mu, \phi\right)\,, & \text{if } y = 1 \\
  \end{array}\right. \, ,
  % \begin{array}{ll}
  %   F_{\rm (B4)}(0 \mid \mu, \phi, -u, 1 + u)\,, & \text{if } y = 0 \\
  %   f_{\rm (B4)}(y \mid \mu, \phi, -u, 1 + u)\,, & \text{if } y \in (0, 1) \\
  %   1 - F_{\rm (B4)}(1 \mid \mu, \phi, -u, 1 + u)\,, & \text{if } y = 1 \\
  % \end{array}\right. \, ,
\end{equation}
where $F_{\rm (B)}(y \mid \mu, \phi)$ is the distribution function of
the beta distribution with density~(\ref{eq:beta}). For the parameter
settings considered in Figure~\ref{fig:xbeta}, the density for the
continuous part of the XB distribution is exactly equal to the dark
grey regions in $(0, 1)$, while the discrete part will consist of
point masses at $0$ and $1$ that are equal to the probability at the
left and right of $(0, 1)$.

A direct corollary of Theorem~\ref{th:beta4} is that the XB
distribution coincides with the beta distribution for $u = 0$, and
that it converges to the normal distribution censored in $[0, 1]$ as
$u \to \infty$. Hence, regression models based on the XB distribution
allow to parameterize the transition from beta regression to the
heteroscedastic two-limit tobit, at the expense of a single extra
parameter~$u$. The beta distribution and the normal distribution
limits of the XB distribution can assume similar shapes in $(0, 1)$
for certain parameter settings. Hence, it is natural to consider
shrinking $u$ to $0$, or, equivalently, shrinking the XB distributions
towards the beta distribution, to avoid any identifiability issues
that may arise in the estimation of the parameters of the XB
distribution, while preserving the flexibility offered by the beta
distribution.

\section{Continuous mixtures of extended-support beta distributions}
\label{sec:xbx}

\subsection{Exponential extended-support beta mixture}

Shrinkage of the XB distribution towards beta can be induced in a
tuning-parameter-free way by constructing a continuous mixture of XB
distributions. In this direction, we consider a continuous mixture of
XB distributions, where $u$ has an exponential distribution with
unknown mean $\nu$, that is $u \sim {\rm Exp}(\nu)$ and
$Y \mid u \sim {\rm XB}(\mu, \phi, u)$.  We denote the marginal
distribution of $Y$ as ${\rm XBX}(\mu, \phi, \nu)$, which has density
\begin{equation}
  \label{eq:xbx}
  f_{\rm (XBX)}(y \mid \mu, \phi, \nu) = \nu^{-1}\int_{0}^\infty
  f_{\rm (XB)}(y \mid \mu, \phi, u) e^{-u/\nu} du \,,
\end{equation}
where $f_{\rm (XB)}(y \mid \mu, \phi, u)$ is defined
in~(\ref{eq:xb}). By its definition, for fixed $\phi$ and $\nu$, the
XBX density is symmetric in that
$f_{\rm (XBX)}(y \mid \mu, \phi, \nu) = f_{\rm (XBX)}(1 - y \mid 1 -
\mu, \phi, \nu)$.  Density~(\ref{eq:xbx}) is not available in closed
form, but it can be accurately and efficiently approximated using a
Gauss-Laguerre quadrature rule
\citep[see,][Equation~25.4.45]{abramowitz+stegun:1964} as
\begin{equation}
  \label{eq:xbx-a}
f_{\rm (XBX)}(y \mid \mu, \phi, \nu) \simeq \sum_{t = 1}^T W_t f_{\rm (XB)} (y \mid \mu, \phi, \nu Q_t) \,,
\end{equation}
where $W_t$ and $Q_t$ $(t = 1, \ldots, T)$ are weights and nodes whose
calculation and derivation are through Laguerre polynomials.

\begin{figure}[t!]
  \begin{center}
    \includegraphics[width = \textwidth]{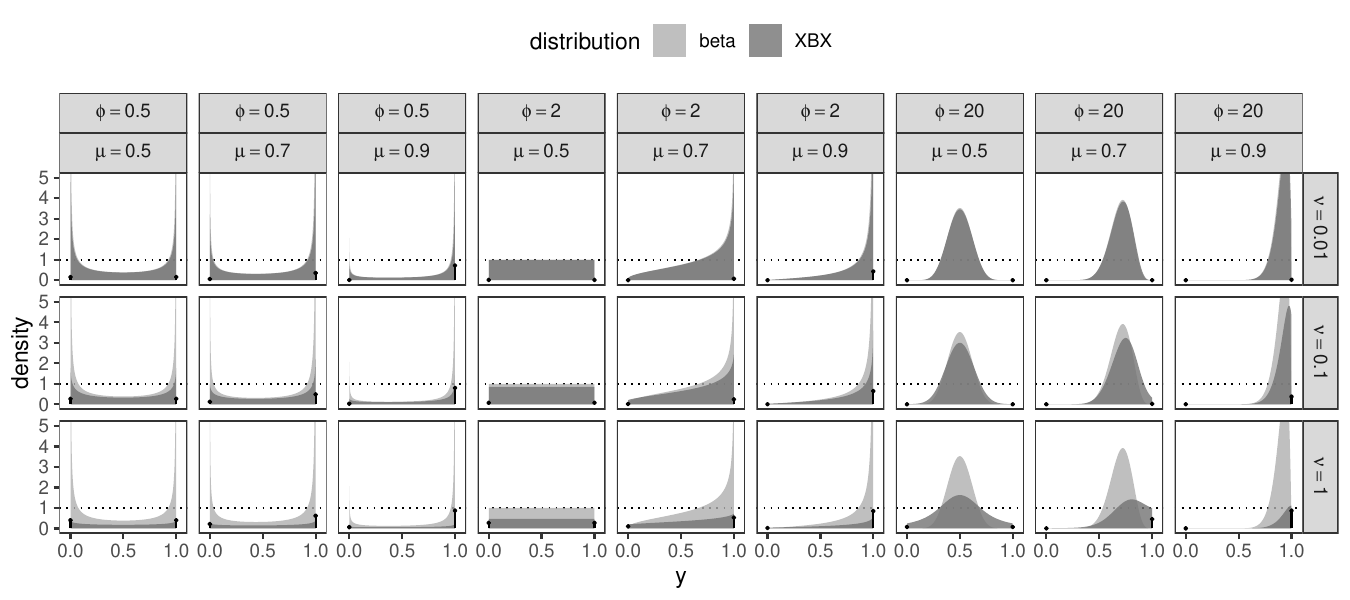}
  \end{center}
  \caption{Densities of the beta and XBX distribution for all combinations of $\mu \in \{0.5, 0.7, 0.9\}$, $\phi \in \{0.5, 2, 20\}$, and $\nu \in \{0.01, 0.1, 1\}$. The horizontal dotted line is at 1. The height of the vertical black segments represents the point mass at $0$ and $1$ for the corresponding parameter setting.}
  \label{fig:xbx}
\end{figure}

Figure~\ref{fig:xbx} shows the density in $(0, 1)$ and the point
masses at $0$ and $1$ of the XBX distribution for a range of values
for $(\mu, \phi, \nu)$. As expected, as $\nu \to 0$, the XBX
distribution converges to the beta distribution. On the other hand, as
$\nu$ increases the point masses at the boundaries increase, and the
relative size of that increase at $0$ and $1$ depends on the skewness
of the XBX distribution.

\subsection{XBX regression}
The operations of mixing and censoring that are involved in the
definition of the XBX distribution are interchangeable. Specifically,
the XBX distribution can be equivalently defined as the distribution
of $Y = \max(\min(Y^*, 1), 0)$, where
$Y^* \mid u \sim {\rm B4}(\mu, \phi, - u, 1 + u)$ , and
$u \sim {\rm Exp}(\nu)$.  The expectation and variance of $Y^* \mid u$
are $\E(Y^* \mid \mu, \phi, u) = (1 + 2 u)\mu - u$ and
$\var(Y^* \mid \mu, \phi, u) = (1 + 2u)^2 \mu (1 - \mu) / (1 + \phi)$,
respectively. The marginal expectation and variance of $Y^*$, when
$u \sim {\rm Exp}(\nu)$, are directly obtained by the laws of total
expectation and variance as
\begin{equation}
  \label{eq:b4mix-mean-var}
  \begin{aligned}
  \E(Y^* \mid \mu, \phi, \nu) & = (1 + 2\nu) \mu - \nu \, , \\ 
  \var(Y^* \mid \mu, \phi, \nu) & = (2 \mu - 1)^2 \nu^2 + (1 + 4\nu + 8 \nu^2)\frac{\mu(1 - \mu)}{1 + \phi} \, ,
\end{aligned}
\end{equation}
respectively. Similarly to the heteroscedastic two-limit tobit, we can
assume that $Y^*$ is a latent variable in $\Re$, which is linked to an
observable response in $[0, 1]$, where the boundary values of $0$ and
$1$ represent too small and too large values of the latent variable.
As discussed in Section~\ref{sec:boundary}, this is a natural
assumption in assessment scores and testing settings, like that of
Example~\ref{ex:reading-skills}.

From expressions~(\ref{eq:b4mix-mean-var}), the parameters $\mu$ and
$\phi$ have direct interpretations in terms of location and precision
of the distribution of $Y^*$, as they do in beta regression. In
particular, the expectation of $Y^*$ is a linear function of $\mu$,
with slope $(1 + 2 \nu)$ and intercept $-\nu$. On the other hand, the
variance of $Y^*$ decreases as $\phi$ increases, with the rate of
decrease increasing quadratically with $\nu$.

So, an XBX regression with mean and dispersion effects for the latent
distribution can directly be defined by linking $\mu$ and $\phi$ to
covariates as in~(\ref{eq:link}). From Figure~\ref{fig:xbx}, it is
directly apparent that the response distribution in XBX regression can
have much more flexible shapes than the censored normal distribution
underlying the heteroscedastic two-limit tobit; see, also
Figure~\ref{fig-supp:xbx-vs-cn} and
Figure~\ref{fig-supp:xbx-vs-cn-hist} in
Section~\ref{sec-supp:xbx-vs-cn} of the supplementary material
document for a side-by-side comparison of the XBX distribution to the
doubly-censored normal distribution.

The expectation and variance of a random variable with XBX
distribution are, of course, nonlinear functions of $\mu$, $\phi$ and
$\nu$ as they depend on the probabilities of $Y^*$ being greater or
equal to $1$ and $Y^*$ being in $(0, 1)$, and on moments of the
truncated distribution of $Y^*$ in $(0,
1)$. Figure~\ref{fig-supp:xbx-e} and Figure~\ref{fig-supp:xbx-v} of
Section~\ref{sec:xbx-e-v} the supplementary material show the
expectation and variance of the XBX distribution as a function of
$\mu$ for a range of values for $\phi$ and $\nu$. For reference,
Figure~\ref{fig-supp:xbx-e} and Figure~\ref{fig-supp:xbx-v} also show
the expectation ($\mu$) and variance ($\mu(1 - \mu) / \phi$) of the
beta distribution as functions of $\mu$ for a range of values of
$\phi$.

As expected, for small $\nu$, the expectation and variance of the XBX
distribution are almost identical to those of the beta
distribution. Notably, for all values of $\nu$ and $\phi$, the
expectation of the XBX distribution is found to be a monotone
increasing function of $\mu$ and rotationally symmetric at
$(0.5, 0.5)$. Hence, any interpretations of the signs of the
coefficients in the regression structure for $\mu$ in XBX regression
are as in standard beta regression, with the usual caveat that changes
in covariates in the regression structure on $\mu$ impact also the
variance of the response distributions. Like for the beta
distribution, the variance of the XBX distribution as a function of
$\mu$ is bounded from above for any $\phi$ and $\nu$, with the
difference that it can assume bimodal shapes. Furthermore, as for
beta, the variance of the XBX distribution decreases as $\phi$
increases. In contrast to beta, however, the mean of the XBX
distribution increases with $\phi$ if $\mu > 0.5$ and decreases if
$\mu < 0.5$, most notably for large values of $\nu$. So, in XBX
regression, if a regression structure on $\phi$ is present alongside a
regression structure on $\mu$, the associated covariates would also
contribute to the mean of the response distribution. The same holds
for other censored regression models, like the two-limit tobit, where
regression structures on the variance of the latent variable,
contribute to the mean of the response distribution.

\subsection{Maximum approximate likelihood}
From equation~(\ref{eq:xbx-a}), the log-likelihood function of the XBX
regression model can be approximated using the Gauss-Laguerre
quadrature rule. The parameters are the regression coefficients~$\beta$
and~$\gamma$ from (\ref{eq:link}) as well as the log
exceedance parameter $\xi = \log\nu$, to assure that all parameters
have unrestricted ranges.  This yields the approximate log-likelihood
\begin{equation}
  \label{eq:loglik-xbx}
  % \ell^{\rm (a)}(\beta, \gamma, \nu) = \sum_{i = 1}^n \log \sum_{t = 1}^T W_t f_{\rm (XB)} (y_i \mid \mu_i, \phi_i, \nu Q_t) \, .
    \ell^{\rm (a)}(\beta, \gamma, \xi) = \sum_{i = 1}^n \log \sum_{t = 1}^T W_t L_i(\beta, \gamma, e^\xi Q_t) \, ,
\end{equation}
where
$L_i(\beta, \gamma, u) = f_{\rm (XB)} (y_i \mid \mu_i, \phi_i, u)$ is
the contribution of the $i$th observation to the likelihood of the
corresponding regression model based on the XB distribution;
see~(\ref{eq:xb}). The gradient of~(\ref{eq:loglik-xbx}) is
\begin{equation}
  \label{eq:score-xbx}
  s^{\rm (a)}(\beta, \gamma, \xi) = \sum_{i = 1}^n \frac{\sum_{t = 1}^T W_t L_i(\beta, \gamma, e^\xi Q_t) \nabla \log L_i(\beta, \gamma, e^\xi Q_t) }{\sum_{t = 1}^T W_t L_i(\beta, \gamma, e^\xi Q_t)} \, .
  % s^{\rm (a)}(\beta, \gamma, \nu) = \sum_{i = 1}^n \frac{\sum_{t = 1}^T W_t f_{\rm (XB)} (y_i \mid \mu_i, \phi_i, \nu Q_t)\nabla \log f_{\rm (XB)} (y_i \mid \mu_i, \phi_i, \nu Q_t) }{\sum_{t = 1}^T W_t f_{\rm (XB)} (y_i \mid \mu_i, \phi_i, \nu Q_t)} \, ,
\end{equation}

Implementing the approximate log-likelihood~(\ref{eq:loglik-xbx}) for
XBX regression is straightforward. The key ingredients are a routine
to compute the weights and nodes of the Gauss-Laguerre quadrature,
such as the \code{gauss.quad()} function of the \pkg{statmod} R package
\citep{smyth:2005}, and routines to compute the density and
distribution functions of beta random variables, such as the
\code{dbeta()} and \code{pbeta()} R functions. From~(\ref{eq:score-xbx}),
the gradient of the approximate log-likelihood for XBX regression is
also readily available, in light of expressions for the contributions
$\nabla \log L_i(\beta, \gamma, \nu Q_t)$ to the log-likelihood
gradient of the corresponding regression model based on the XB
distribution. A tedious but straightforward calculation gives that the
gradient of $\log L_i(\beta, \gamma, u)$ with respect to $\beta$ is
\begin{equation}
  \label{eq:score-xbx-beta}
  \nabla_\beta \log L_i(\beta, \gamma, u) = \left\{
  \begin{array}{ll}
    \displaystyle x_i^\top \phi_i d_{1, i} \frac{Q_{0, i} - R_{0, i}}{F_{0, i}} \,, & \text{if } y = 0 \\[1em]
    \displaystyle x_i^\top \phi_i d_{1, i} (\bar{T}_i - \bar{U}_i) \,, & \text{if } y \in (0, 1) \\[1em]
    \displaystyle x_i^\top \phi_i d_{1, i} \frac{R_{1, i} - Q_{1, i}}{1 - F_{1, i}}  \,, & \text{if } y = 1
  \end{array}\right.\, ,
\end{equation}
the gradient with respect to $\gamma$ is 
\begin{equation}
  \label{eq:score-xbx-gamma}
  \nabla_\gamma \log L_i(\beta, \gamma, u) = \left\{
  \begin{array}{ll}
    \displaystyle z_i^\top d_{2, i} \frac{\mu_i Q_{0, i} + (1 - \mu_i)R_{0, i}}{F_{0, i}} \,, & \text{if } y = 0 \\[1em]
    \displaystyle z_i^\top d_{2, i} \{\mu_i (\bar{T}_i - \bar{U}_i) + \bar{U}_i\} \,, & \text{if } y \in (0, 1) \\[1em]
    \displaystyle z_i^\top d_{2, i} \frac{\mu_i Q_{1, i} + (1 - \mu_i)R_{1, i}}{1 - F_{1, i}} \,, & \text{if } y = 1
  \end{array}\right.\, ,
\end{equation}
and that the partial derivative with respect to $u$ is 
\begin{equation}
  \label{eq:score-xbx-u}
  \frac{\partial}{\partial u} \log L_i(\beta, \gamma, u) = \left\{
  \begin{array}{ll}
    \displaystyle \frac{f_{0, i}}{(1 + 2u)^2F_{0, i}} \,, & \text{if } y = 0 \\[1em]
    \displaystyle \frac{\mu_i \phi_i - 1}{y_i + u} + \frac{(1 - \mu_i) \phi_i - 1}{1 - y_i + u} - \frac{2 (\phi_i - 1)}{1 + 2u} \,, & \text{if } y \in (0, 1) \\[1em]
    \displaystyle \frac{f_{1, i}}{(1 + 2u)^2(1 - F_{1, i})} \,, & \text{if } y = 1
  \end{array}\right.\, .
\end{equation}
In expressions~(\ref{eq:score-xbx-beta}),~(\ref{eq:score-xbx-gamma}),
and~(\ref{eq:score-xbx-u}), $d_{1, i} = d \mu_i / d \eta_i$,
$d_{2, i} = d \phi_i / d \zeta_i$, and 
\[
  f_{j, i} = f_{(b)}(z(j, u) \mid \mu_i, \phi_i) \quad \text{and} \quad F_{j, i} = F_{(b)}(z(j, u) \mid \mu_i, \phi_i) \,,
\]
% \begin{align*}
%   f_{j, i} & = f_{(b)}(z(j, u) \mid \mu_i, \phi_i) \,, \\
%   F_{j, i} & = F_{(b)}(z(j, u) \mid \mu_i, \phi_i) \, ,
% \end{align*}
with $z(j, u) = (j + u)/(1 + 2u)$. In addition,
\[
  \bar{T}_i = \log Y_i - \psi(\phi\mu_i) + \psi(\phi_i) \quad \text{and} \quad
  \bar{U}_i = \log (1 - Y_i) - \psi(\phi(1 - \mu_i)) + \psi(\phi_i) \, ,
\]
% \begin{align*}
% \bar{T}_i & = \log Y_i - \psi(\phi\mu_i) + \psi(\phi_i) \, , \\
% \bar{U}_i & = \log (1 - Y_i) - \psi(\phi(1 - \mu_i)) + \psi(\phi_i) \, ,
% \end{align*}
$Q_{j, i} = q(z(j, u) \mid \mu_i, \phi_i)$, and $R_{j, i} = R(z(j, u) \mid \mu_i, \phi_i)$ with
\begin{align*}
  q(z \mid \mu, \phi) & = F_{(B)}(z \mid \mu, \phi)\left\{
    \psi(\phi) - \psi(\phi\mu) + \log z \right\} - h(z
                          \mid \mu, \phi)\,, \\
  r(z \mid \mu, \phi) & = F_{(B)}(1 - z \mid 1 - \mu, \phi)\left\{
    \psi(\phi(1- \mu)) - \psi(\phi) - \log(1 - z) \right\}
  - h(1 - z
  \mid 1 - \mu, \phi)\, ,
\end{align*}
and
\[
  h(z, \mu, \phi) = \frac{z^{\mu\phi}}{\mu^2\phi^2B(\mu\phi, (1 - \mu)\phi)} {_3}F_{2}(\mu\phi, \mu\phi, 1 - (1 -\mu)\phi \,;\,
  \mu\phi + 1, \mu\phi + 1 \,;\, z) \, ,
\]
where $\psi(\cdot)$ is the digamma function, and
${}_3F_{2}(\cdot, \cdot, \cdot \,;\, \cdot, \cdot)$ is the generalized
hypergeometric function.

Maximum approximate likelihood estimates can be computed using either
gradient-free optimization algorithms or optimization algorithms that
operate with numerical or analytical gradient
approximations. Estimated standard errors can be computed as the
diagonal of the inverse of the negative hessian of the approximate
log-likelihood, which is the observed information matrix, and can be
accurately approximated using numerical differentiation procedures
on~(\ref{eq:loglik-xbx}) or~(\ref{eq:score-xbx}).

\subsection{Inference}

Hypothesis tests and confidence sets can be constructed using standard
tools for likelihood inference. For example, letting
$\theta = (\beta^\top, \gamma^\top)^\top$ and $I^{(a)}(\hat\theta)$ be
the negative hessian of the approximate
log-likelihood~(\ref{eq:loglik-xbx}), the hypothesis $R \theta = b$
for a known $c \times (p + q)$ martix $R$ can be tested by comparing
the approximate likelihood ratio statistic
\[
  2 \left \{\ell^{(a)}(\hat\beta, \hat\gamma, \hat\xi) - \max_{R\theta = b} \ell^{(a)}(\beta, \gamma, \xi)\right\} \,,
\]
or the Wald statistic
\[
  (R\hat\theta - b)^\top \{ R I^{(a)}(\hat\theta)^{-1} R^\top \}^{-1} (R\hat\theta - b) \,,
\]
to the quantiles of their asymptotic distribution, that is $\chi^2$
with $c$ degrees of freedom.

\section{Application: Loss aversion in investment decisions}
\label{sec:lossaversion}

\subsection{Motivation}

To illustrate the benefits of XBX regression models, we revisit the
analysis of a behavioral economics experiment conducted and published
by \cite{glaetzleruetzler+sutter+zeileis:2015}.  The outcome variable
is the proportion of tokens invested by high-school students in a
risky lottery with positive expected
payouts. \cite{glaetzleruetzler+sutter+zeileis:2015} focused on the
effects of several experimental factors on the mean investments, which
reflect the players' willingness to take risks. The latter study
employed linear regression models, estimated by ordinary least squares
(OLS) with standard errors adjusted for potential clustering and
heteroscedasticity.

Here, we extend the analysis from
\cite{glaetzleruetzler+sutter+zeileis:2015} by employing a similar
model for the mean investments but exploring distributional
specifications that allow for a probabilistic, rather than mean-only,
interpretation of the effects.  From an economic perspective this is
of interest because it allows to interpret both the mean willingness
to take risks in this experiment, and the probability to behave like a
rational \emph{Homo oeconomicus}, who would invest (almost) all tokens
in this lottery because it has positive expected payouts.
% From a statistical perspective the latter corresponds to
% $P(Y \ge 1 - \epsilon)$ with some small margin $\epsilon$

\subsection{Data}
\label{sec:lossaversion-data}

The experiment was designed by Matthias Sutter and Daniela
Gl\"atzle-R\"utzler at Universit\"at Innsbruck to study a phenomenon
known as myopic loss aversion \citep{benartzi+thaler:1995}.
The classical setup for demonstrating this is a lottery with
positive expected payouts where individuals typically invest less than
the maximum possible (loss aversion), and where investments further
decrease if the individuals can make (many) short-term investment decisions
(myopia or short-sightedness) as opposed to (fewer) long-term
decisions.

The experiment involved high-school students from Schwaz and Innsbruck
in Tyrol, Austria, who participated as individuals or as teams of two
students. The subjects, consisting of 385~individuals and 185~teams of
two, could invest $X \in \{0, \dots, 100\}$ tokens in each of nine
rounds in a lottery. The payouts were $100 + 5 X / 2$ tokens with
probability $1/3$ and $100 - X$ tokens with probability $2/3$. Thus,
the expected payouts were $100 + X / 6$ tokens, and, as a result, a
\emph{Homo oeconomicus} would always invest the maximum of $X = 100$
tokens in each round. It was observed that 1.5\% of the subjects never
invested anything, and 5.3\% always invested everything.

Interest is in linking the proportion $Y_i \in [0, 1]$ of the total
tokens invested by the $i$th subject in the nine rounds out of the
$900$ available tokens to explanatory information. The explanatory
information we consider here consists of the subject characteristics
and their interactions that
\cite{glaetzleruetzler+sutter+zeileis:2015} found to be relevant in
the experiment. These are \emph{Grade} ($G_{i}$ with value 0 and 1,
according to whether the the $i$th subject is from lower grades 6--8
or upper grades 10--12, respectively), \emph{Arrangement} ($T_{i}$,
with values 0 and 1, according to whether the $i$th subject is an
individual or a team of two, respectively), \emph{Age} ($A_{i}$ the age of
the player or the average age of the team of players in years),
\emph{Sex} ($S_i$, with value 1 if the player was male, or the team
included at least one male student, and 0, otherwise), and the
interactions of \emph{Grade} with \emph{Arrangement} and \emph{Age}.

\subsection{Models}
\label{sec:la-models}

We compare four different models for the proportion of the total
tokens invested. All models assume that $Y_1, \ldots, Y_n$ $(n = 570)$
are independent conditionally on the covariate information, and that
the mean parameter for the distribution of $Y_i$ is linked to a linear
predictor of the form
\begin{equation}
  \label{eq:la_mean}
  \eta_i = \beta_1 + \beta_2 G_i + \beta_3 T_i + \beta_4 A_i + \beta_5 S_i + \beta_5 G_iT_i + \beta_7 G_i A_i \, .
\end{equation}
The first model is a normal linear regression model (N) under which $Y_i$ has a
normal distribution with mean $\mu_i = \eta_i$ and variance
$\sigma_i = \sigma^2$.  This model corresponds to the OLS approach
used by \cite{glaetzleruetzler+sutter+zeileis:2015}. The second model
we consider is the heteroscedastic two-limit tobit model, which is
based on the censored normal distribution (CN) and described in
Section~\ref{sec:boundary}, and where the latent variable $Y_i^*$ has
mean $\mu_i = \eta_i$ and variance $\sigma_i$ satisfying
$\log(\sigma_i) = \zeta_i$ with
\begin{equation}
  \label{eq:la_var}
  \zeta_i = \gamma_1 + \gamma_2 G_i + \gamma_3 T_i + \gamma_4 S_i \, .
\end{equation}
For simplicity, the regression structure~(\ref{eq:la_var}) involves
only main effects, excluding $A_i$ because $G_i$ already captures the
main age differences between the lower and upper grades. The third
model we consider is the beta regression model (B) after adjusting the
responses using the ad-hoc scaling of \cite{smithson+verkuilen:2006}
(as described in Section~\ref{sec:boundary}) with mean $\mu_i$
satisfying $\log \{\mu_i / (1 - \mu_i) \} = \eta_i$, with $\eta_i$ as
in~(\ref{eq:la_mean}), and precision parameter $\phi_i$ satisfying
$\log \phi_i = \zeta_i$, with $\zeta_i$ as
in~(\ref{eq:la_var}). Finally, we also consider the extended-support
beta mixture model (XBX) where the mean and precision parameters
$\mu_i$ and $\phi_i$ (see~(\ref{eq:b4mix-mean-var})), respectively,
have the same specification as in the beta regression model.

\subsection{Estimates}

Coefficient estimates (and estimated standard errors) along with
log-likelihood, Akaike and Bayes information criteria (AIC and BIC)
for all four models are provided in
Table~\ref{tab:lossaversion-gaussian} and
Table~\ref{tab:lossaversion-beta}, respectively. Separate tables are
used because the coefficients from N and CN use an identity link for
the mean parameter whereas B and XBX use a logit link. In addition,
the log-likelihood, and, hence, AIC and BIC, are comparable only
between CN and XBX because those two models have the same support for
the response variable, that is the unit interval with point masses at
0 and 1.

\begin{table}[t!]
  \caption{Coefficient estimates, estimated standard errors (in
    parentheses), maximized log-likelihood, AIC, and BIC for the
    normal model (N) and the censored normal model (CN) of
    Section~\ref{sec:la-models}.}
\sisetup{input-open-uncertainty = ,
         input-close-uncertainty = ,
         table-align-text-pre = false,
         table-align-text-post = false}
\begin{center}
\begin{tabular*}{\textwidth}{l@{\extracolsep{\fill}}*{4}{S[table-format=-1.3]@{\,}}}
\toprule
Covariate                    & \multicolumn{2}{c}{N}  & \multicolumn{2}{c}{CN} \\ \cmidrule{2-3} \cmidrule{4-5} 
                             & \multicolumn{1}{c}{$\mu_i$}       & \multicolumn{1}{c}{$\sigma_i$}   & \multicolumn{1}{c}{$\mu_i$}      & \multicolumn{1}{c}{$\log(\sigma_i)$} \\ \midrule
(Intercept)                  &  0.284  & 0.246 &  0.250  &  -1.467 \\
                            & (0.157) &             & (0.154) & (0.047) \\ \addlinespace 
$G_i$ (Grade)                &  -0.844 &             & -0.982  &  0.167  \\ 
                            & (0.281) &             & (0.317) & (0.066) \\ \addlinespace 
$T_i$ (Arrangement)          &  0.063  &             &  0.065  &  -0.138 \\
                            & (0.030) &             & (0.029) & (0.072) \\ \addlinespace 
$A_i$ (Age)                  &  0.012  &             &  0.014  &  \\
                            & (0.012) &             & (0.012) & \\ \addlinespace 
$S_i$ (Sex)                  &  0.104  &             &  0.113  &  0.212 \\
                            & (0.023) &             & (0.025) & (0.068) \\ \addlinespace 
$G_iT_i$ (Grade $\times$ Arrangement)   &  0.151  &             &  0.167  & \\
                            & (0.045) &             & (0.050) & \\ \addlinespace 
$G_iA_i$ (Grade $\times$ Age)                    &  0.046  &             &  0.054  &               \\
                            & (0.019) &             & (0.020) &               \\ \midrule
Log-likelihood               & \multicolumn{2}{r}{$-8.8$}  & \multicolumn{2}{r}{$-75.5$}      \\
AIC                          & \multicolumn{2}{r}{$33.6$}  & \multicolumn{2}{r}{$173.0$}      \\
BIC                          & \multicolumn{2}{r}{$68.3$}  & \multicolumn{2}{r}{$220.8$}      \\ \bottomrule
\end{tabular*}
\end{center}
\label{tab:lossaversion-gaussian}
\end{table}

The coefficient estimates (and the estimated standard errors) in the
$\mu_i$ regression structure for the N and CN models are similar, with
the ones for N being slightly smaller than those for CN. This
resembles the attenuation bias phenomenon when estimating regression
parameters using least squares in the presence of bounded responses
with boundary observations \citep[see, for
example,][Section~7.2]{winkelmann+boes:2006}.

\begin{table}[t!]
  \caption{Coefficient estimates, estimated standard errors (in
    parentheses), maximized log-likelihood, AIC, and BIC for the beta
    model (B) and the extended beta mixture model (XBX) of
    Section~\ref{sec:la-models}.}
\sisetup{input-open-uncertainty = ,
         input-close-uncertainty = ,
         table-align-text-pre = false,
         table-align-text-post = false}
\begin{center}
\begin{tabular*}{\textwidth}{l@{\extracolsep{\fill}}*{5}{S[table-format=-1.3]@{\,}}}
\toprule
Covariate & \multicolumn{2}{c}{B} & \multicolumn{3}{c}{XBX} \\ \cmidrule{2-3} \cmidrule{4-6}
          & \multicolumn{1}{c}{$\text{logit}(\mu_i)$} & \multicolumn{1}{c}{$\log(\phi_i)$}  & \multicolumn{1}{c}{$\text{logit}(\mu_i)$} & \multicolumn{1}{c}{$\log(\phi_i)$} & \multicolumn{1}{c}{$\log(\nu)$} \\ \midrule
  (Intercept)                  &  -1.414 &  1.194  &  -0.865 & 1.756   &  -2.273 \\
                             & (0.620) & (0.084) & (0.519) & (0.128) & (0.245) \\  \addlinespace 
$G_i$ (Grade)               &  -2.943 &  -0.553 &  -3.096 &  -0.316 & \\
                             & (1.252) & (0.104) & (1.053) & (0.131) & \\  \addlinespace 
$T_i$ (Arrangement)            &  0.225  &  0.406  &  0.208  &  0.325  & \\
                             & (0.117) & (0.122) & (0.099) & (0.145) & \\  \addlinespace 
$A_i$ (Age)                  &  0.091  &               &  0.049  & & \\
                             & (0.049) &               & (0.041) & &  \\  \addlinespace 
$S_i$ (Sex)                  &  0.455  &  -0.555 &  0.379  & -0.484  & \\
                             & (0.099) & (0.113) & (0.084) & (0.136) & \\  \addlinespace 
$G_iT_i$ (Grade $\times$ Arrangement) &  0.655  &               &  0.567  &  & \\
                             & (0.200) &               & (0.169) &  & \\  \addlinespace 
$G_iA_i$ (Grade $\times$ Age)  &  0.151  &               &  0.169  &  & \\
                             & (0.081) &               & (0.068) &  & \\ \midrule
Log-likelihood               & \multicolumn{2}{r}{$94.4$}     & \multicolumn{3}{r}{$-71.8$}  \\
AIC                          & \multicolumn{2}{r}{$-166.8$}  & \multicolumn{3}{r}{$167.7$}  \\
BIC                          & \multicolumn{2}{r}{$-119.0$} & \multicolumn{3}{r}{$219.8$}  \\ \bottomrule
\end{tabular*}
\end{center}
\label{tab:lossaversion-beta}
\end{table}

Similarly, the coefficient estimates in the $\log \{ \mu_i / (1 - \mu_i) \}$
regression structure for the B and XBX are comparable and lead to
similar qualitative interpretations. The estimates for XBX (and their
estimated standard errors) are slightly smaller than for B. The reason
is that the estimated effects in B are slightly exaggerated to push
fitted values towards the boundaries, while the extra parameter $\nu$,
estimated to be $0.1$, in XBX offers more flexibility in capturing the
boundary observations.

\begin{figure}[t!]
  \begin{center}
    \includegraphics[width = \textwidth]{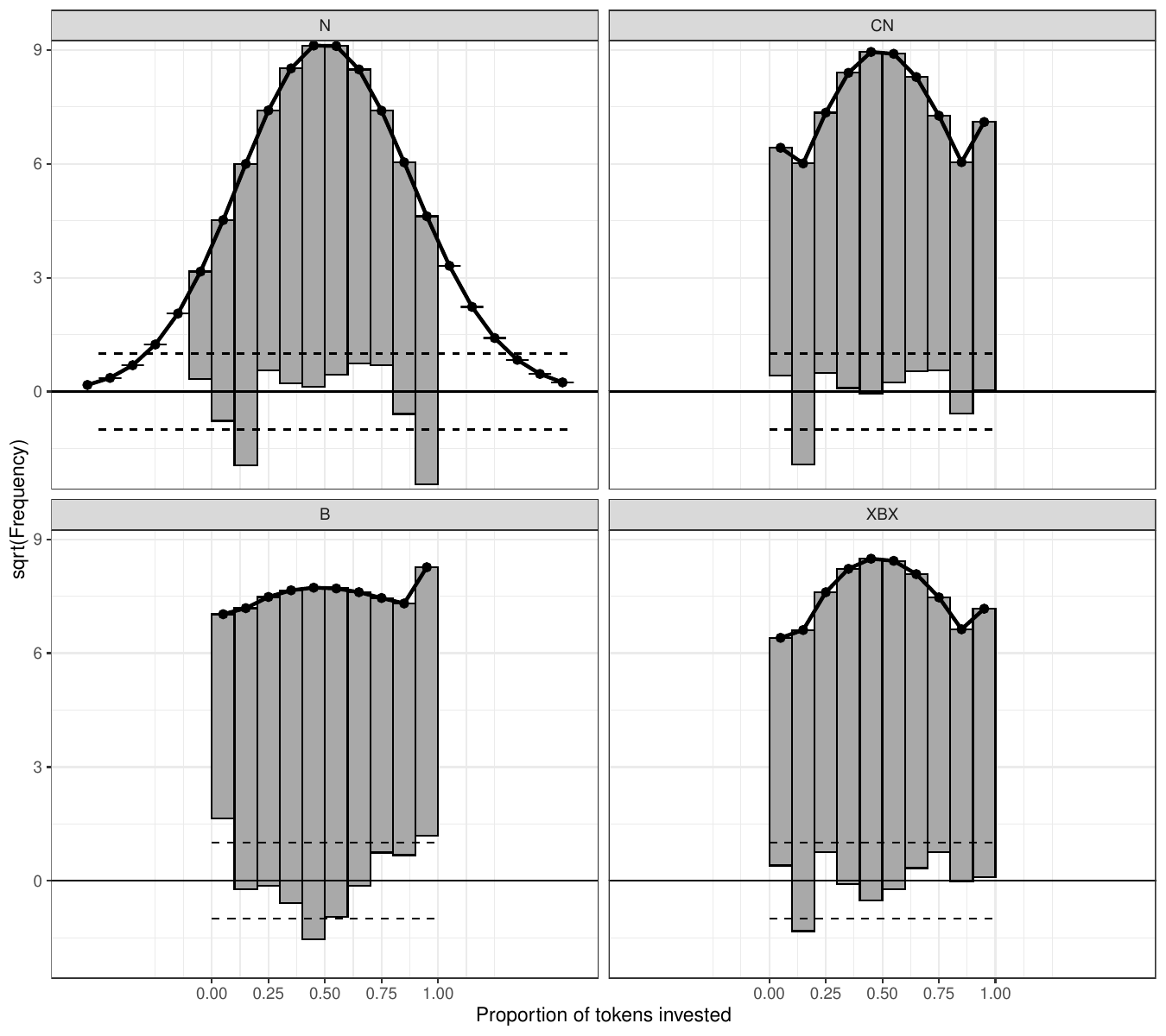}
  \end{center}
  \caption{Hanging rootograms for the fitted models in
    Table~\ref{tab:lossaversion-gaussian} and
    Table~\ref{tab:lossaversion-beta}. The rootograms compare the
    histogram estimate of the marginal distribution of the proportions
    of tokens invested with the corresponding expected frequencies
    from the fitted models. The square root of the expected
    frequencies are shown as dots connected by a solid line (in black) and the square root of the
    observed frequencies are hanging from the points as grey bars. The
    quality of the fit is determined by checking whether the bars hang
    down to the horizontal line at zero or whether they are clearly
    too large or too small. The dashed horizontal lines are the Tukey warning
    limits at $\pm 1$ \citep[see][for details]{kleiber+zeileis:2016}.}
  \label{fig:lossaversion-rootogram}
\end{figure}

\subsection{Assessment of fitted models}
\label{sec:la-fit}
The differences between the models become more apparent when
considering the fitted
distributions. Figure~\ref{fig:lossaversion-rootogram} shows hanging
rootograms \citep[see][for details]{kleiber+zeileis:2016} for the
fitted models in Table~\ref{tab:lossaversion-gaussian} and
Table~\ref{tab:lossaversion-beta}. The rootograms compare the
empirical marginal distributions of the proportion of tokens invested
to the implied distributions from the fitted models. The quality of
the fit is determined by checking whether the bars hang down to the
horizontal line at zero or whether they are clearly too large or too
small. As is apparent, models N and B fit poorly in the tails. For N
the support of the fitted distributions is the real line and, as a
result, the expected frequencies in intervals close to $0$ or $1$ are
larger and smaller, respectively, than what is observed. The
corresponding expected frequencies for B are larger than what is
observed, while the expected frequencies at intervals around 0.5 are
too small compared to what is observed. 
% , because beta regression struggles with giving enough probability
% mass to intervals close to $0$ or $1$ because it cannot assign point
% masses at the boundaries.  Also, the expected frequencies at
% intervals around 0.5 are too small compared to what is observed.
In contrast, models CN and
XBX fit very well with all but the second bar from the left hanging
close to the zero reference line and within the Tukey limits. The
misfit of the second bar is larger for CN compared to XBX, which is
also reflected by both AIC and BIC being smaller for XBX; see
Table~\ref{tab:lossaversion-gaussian} and
Table~\ref{tab:lossaversion-beta}.

Testing whether the removal of dispersion effects
$\log(\phi) = \zeta_i$ in model XBX, with $\zeta_i$ as
in~(\ref{eq:la_var}), is worthwhile returns an approximate likelihood
ratio statistic of $20.614$ and a Wald statistic of $21.251$ on $3$
degrees of freedom, which is strong evidence against the hypothesis
that $\log(\phi) = \gamma_1$. We get the same qualitative conclusions
from model CN about the removal of variance effects
$\log(\sigma) = \zeta_i$, with a likelihood ratio statistic of
$19.620$ and a Wald statistic of $19.243$ on $3$ degrees of freedom.

\subsection{Effects}

Models XBX and CN allow more economically relevant interpretations of
probability effects than the N and B models. For illustration, we
focus on the team arrangement effect for a subsample with a large
share of very rational subjects: male players or teams with at least
one male, in grades~10--12, and between 15 and 17 years of age. About
$52\%$ of the subjects in that subsample invested more than half of
their total tokens.

\begin{table}[t!]
  \caption{Estimated arrangement effects for subjects that are male or
    teams with at least one male, in grade~10--12 and 16 years of
    age. The rows for $E(Y)$ show model-based estimates of the
    expected proportion of tokens invested, and the rows for
    $P(Y > 0.95)$ show model-based estimates of the probability to
    behave very rationally by investing more than $95\%$ of the
    tokens. All effects are contrasted to the corresponding empirical
    quantities obtained from the subsample of subjects that are male
    or teams with at least one male, in grades~10--12, and between 15
    and 17~years of age. The rows for the parameters show the
    estimates of the distributional parameters for each arrangement
    setting.}
  \sisetup{input-open-uncertainty = , input-close-uncertainty = ,
    table-align-text-pre = false, table-align-text-post = false}
\begin{center}
\begin{tabular*}{\textwidth}{l@{\extracolsep{\fill}} l@{\extracolsep{\fill}} *{1}{S[table-format=-1.3]@{\,}} l@{\extracolsep{\fill}} *{2}{S[table-format=-1.3]@{\,}} l@{\extracolsep{\fill}} *{2}{S[table-format=-1.3]@{\,}}}
  \toprule
 & \multicolumn{1}{l}{Arrangement} & \multicolumn{1}{c}{Empirical}   & &  \multicolumn{1}{c}{N}  &  \multicolumn{1}{c}{CN}  & &  \multicolumn{1}{c}{B}  &  \multicolumn{1}{c}{XBX}  \\ \midrule
\multirow{2}{*}{$E(Y)$} & Individual        & 0.482 &  &    0.461 &      0.470 &  &  0.492 &    0.471 \\
 & Team              & 0.665 &  &    0.675 &      0.674 &  &  0.700 &    0.686 \\ \midrule
\multirow{2}{*}{$P(Y > 0.95)$} & Individual        & 0.087 &  &    0.023 &      0.075 &  &  0.125 &    0.072 \\
 & Team              & 0.207 &  &    0.131 &      0.194 &  &  0.249 &    0.185 \\ \midrule
  \multirow{6}{*}{Parameters} & Individual  & &     $\mu$     &  0.461 & 0.465 & $\mu$  & 0.492 &  0.475 \\
 &             &&     $\sigma$  &  0.246 & 0.337 & $\phi$ & 1.091 &  2.602 \\
 &            &&               &        &       & $\nu$  &       &  0.103 \\ \cmidrule{2-9}
 & Team       &  &     $\mu$     &  0.675 & 0.696 & $\mu$  & 0.700 &  0.662 \\
 &            &&     $\sigma$  &  0.246 & 0.294 & $\phi$ & 1.637 &  3.602 \\
 &            &&               &        &       & $\nu$  &       &  0.103 \\ \bottomrule 
\end{tabular*}
\end{center}
\label{tab:lossaversion-effects}
\end{table}

\begin{figure}[t!]
  \begin{center}
    \includegraphics[width = \textwidth]{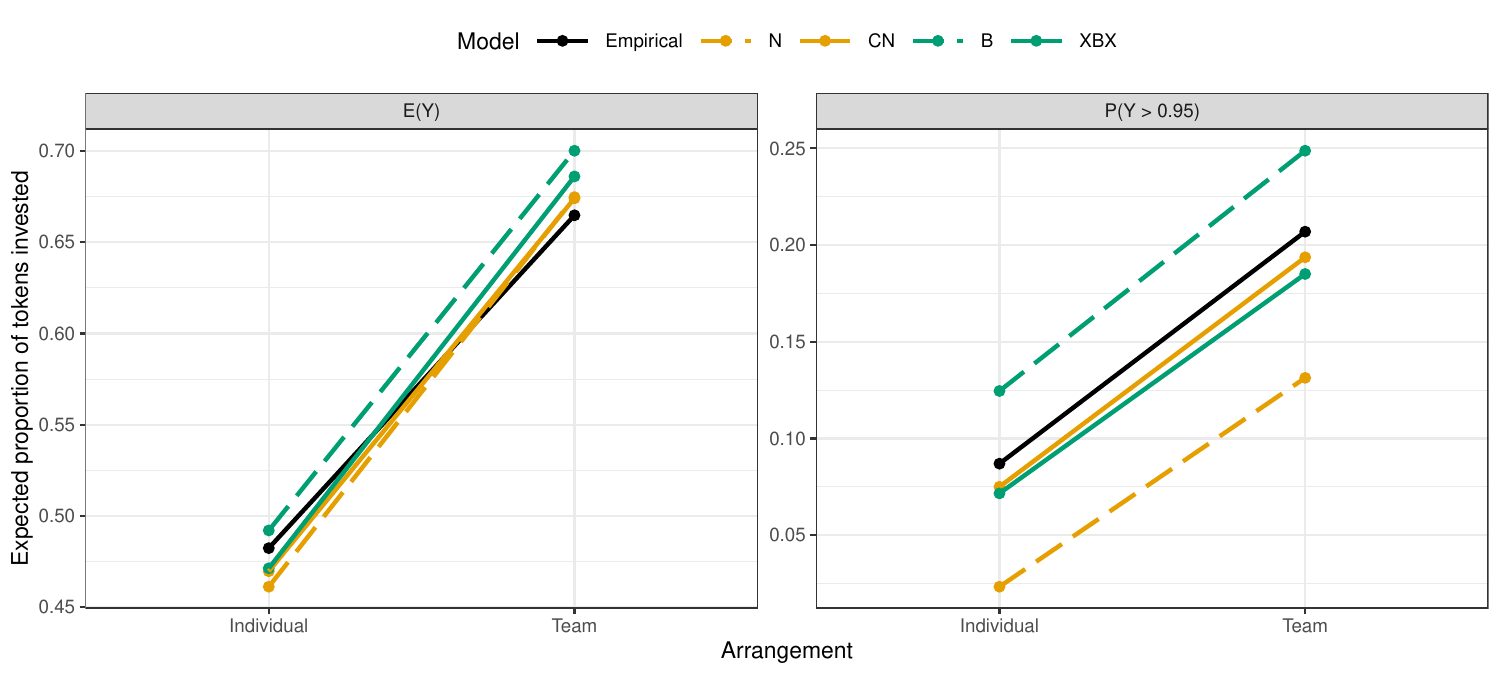}
  \end{center}
  \caption{Estimated arrangement effects in Table~\ref{tab:lossaversion-effects}.}
  \label{fig:lossaversion-meanprob}
\end{figure}

Table~\ref{tab:lossaversion-effects} shows the model-based estimates
of the expected proportion of tokens invested ($E(Y)$) and the
probability to behave very rationally by investing more than $95\%$
($P(Y > 0.95)$) of the tokens, contrasted to the corresponding
empirical quantities from the subsample with age between 15 and 17
years. The mean and probability effects are also visualized in
Figure~\ref{fig:lossaversion-meanprob}; see, also,
Figure~\ref{fig:lossaversion-cdf} in the supplementary material
document for the fitted cumulative distribution functions. All models
do a reasonable job in estimating $E(Y)$, with the largest deviations
observed for model B. Nevertheless, the censored models XBX and CN
estimate the probabilities to behave very rationally much better than
B and N do. Model~N, due to its unlimited support, underestimates the
variance and, hence, underestimates that probability, while model~B
with limited support and no point masses overestimates that
probability. As is apparent using a model with a good fit is essential
for estimating probability effects, especially so at the tails of
bounded domains. Table~\ref{tab:lossaversion-left-prob} in the
supplementary material document shows the model-based estimates for
$P(Y < 0.05)$, which, albeit not being as important as $P(Y > 0.95)$
in terms of its economic interpretation, further illustrates the
quality of the estimated probability effects from XBX and CN at the
tails.

\subsection{Three-part hurdle models}
\label{sec:la-hurdle}

In our analysis, we have not considered the three-part hurdle model of
\citet{ospina+ferrari:2010} (``zero-and-one-inflated'' beta) for two
reasons: (a) from an economic point of view, it is not plausible that
the effects driving ``fully rational'' behavior ($Y = 1$) are
different from the effects driving ``very rational'' behavior (say,
$Y < 0.95$ or $Y = 1$); (b) the hurdle needs substantially more
parameters than XBX due to the regression structures for modelling the
probability of boundary and non-boundary observations, whose
estimation may also not be reliable. For example, for the loss
aversion data, employing a multinomial logistic regression model for
the probabilities of boundary and non-boundary observations with the
same covariates as those in the linear predictor~(\ref{eq:la_mean})
or~(\ref{eq:la_var}), and model B for the non-boundary observations,
results in a three-part hurdle model with 25 or 19 parameters,
respectively, instead of the 12 parameters that XBX requires. Notably,
maximum likelihood estimation results in infinite estimates for the
parameters of either multinomial logistic regression model (see
Table~\ref{tab:lossaversion-beta01-1} and
Table~\ref{tab:lossaversion-beta01-2} in the supplementary material
document for estimates and estimated standard errors) due to data
separation.  Also, the BIC values from both three-part hurdle model
are larger than those of the CN and XBX models.

\section{XBX versus heteroscedastic two-limit tobit regression}
\label{sec:xbx-vs-htobit}

In the application of Section~\ref{sec:lossaversion}, XBX regression
resulted in a slightly better fit than the CN model. In this section,
we carry out a numerical study to compare the performance of XBX and
CN regression in a more systematic way.

Theorem~\ref{th:beta4} shows that regression models based on the XB
distribution allow to parameterize the transition from the beta
distribution to the censored normal distribution with a single
parameter $u$. On the other hand, XBX regression is a continuous
mixture of XB distributions that shrinks towards a beta regression
model. Hence, a comparison of XBX regression with heteroscedastic
two-limit tobit regression can be based on data sets from the XB
distributions. For a given
$(\beta_1, \beta_2, \gamma_1, \gamma_2, u)^\top$, we let
$\log \{ \mu_i / (1 - \mu_i) \} = \beta_1 + \beta_2 x_i$ and
$\log \phi_i = \gamma_1 + \gamma_2 x_i$, where
$x_i = 2 (i - 1) / (n - 1) - 1 \in [-1, 1]$ $(i = 1, \ldots, n)$, and
simulate response values $y_1, \ldots, y_n, y_1', \ldots, y_n'$ as
realizations of independent random variables
$Y_1, \ldots, Y_n, Y_1', \ldots, Y_n'$ with
$Y_i \sim {\rm XB}(\mu_i, \phi_i, u)$, and
$Y_i' \sim {\rm XB}(\mu_i, \phi_i, u)$. We then fit the XBX and CN
models on $y_1, \ldots, y_n$ with covariates $x_1, \ldots, x_n$ in
both mean and variance and precision regression specifications. The
predictive performance of each model is evaluated as
\begin{equation}
  \label{eq:sum-crps}
  S = \sum_{i = 1}^n C(\hat{F}_i, y_i')\,,
\end{equation}
where $\hat{F}_i(\cdot)$ is the model's cumulative distribution
function at $x_i$ at the parameter estimates from $y_1, \ldots, y_n$,
and $C(F, z) = \int_\Re \left\{ F(t) - I(t \ge z)\right\}^2dt$ is the
continuous ranked probability score \citep{gneiting+raftery:2007},
which is a strictly proper scoring rule.

\begin{table}
  \caption{Intervals $[\mu_1, \mu_n]$ and $[\phi_1, \phi_n]$ considered in the simulation study of Section~\ref{sec:xbx-vs-htobit}, and the implied values of $\beta_1$, $\beta_2$, $\gamma_1$ and $\gamma_2$ (in three decimal places).}
  \begin{minipage}[t]{.5\linewidth}
    \vspace{0pt}
    \begin{center}
      \begin{tabular}{rrrr}
        \toprule
        \multicolumn{1}{c}{$\mu_1$} & \multicolumn{1}{c}{$\mu_n$} & \multicolumn{1}{c}{$\beta_1$} & \multicolumn{1}{c}{$\beta_2$} \\
        \midrule
        $0.05$ & $0.25$ & $-2.022$ & $0.923$ \\
        $0.05$ & $0.50$ & $-1.472$ & $1.472$ \\
        $0.05$ & $0.75$ & $-0.923$ & $2.022$ \\
        $0.05$ & $0.95$ & $-0.000$ & $2.944$ \\
        $0.25$ & $0.50$ & $-0.549$ & $0.549$ \\
        $0.25$ & $0.75$ &  $0.000$ & $1.099$ \\
        $0.25$ & $0.95$ &  $0.923$ & $2.022$ \\
        $0.50$ & $0.75$ &  $0.549$ & $0.549$ \\
        $0.50$ & $0.95$ &  $1.472$ & $1.472$ \\
        $0.75$ & $0.95$ &  $2.022$ & $0.923$ \\
        \bottomrule
      \end{tabular}
    \end{center}
  \end{minipage}
  \begin{minipage}[t]{.5\linewidth}
    \vspace{0pt}
    \begin{center}
      \begin{tabular}{rrrr}
        \toprule
        \multicolumn{1}{c}{$\phi_1$} & \multicolumn{1}{c}{$\phi_n$} & \multicolumn{1}{c}{$\gamma_1$} & \multicolumn{1}{c}{$\gamma_2$} \\
        \midrule
         $0.5$ &  $10.0$ & $0.805$ & $1.498$ \\
         $0.5$ &  $20.0$ & $1.151$ & $1.844$ \\
         $0.5$ &  $50.0$ & $1.609$ & $2.303$ \\
         $5.0$ & $100.0$ & $3.107$ & $1.498$ \\
        $10.0$ &  $20.0$ & $2.649$ & $0.347$ \\
        $20.0$ &  $50.0$ & $3.454$ & $0.458$ \\
        $50.0$ & $100.0$ & $4.259$ & $0.347$ \\
        \bottomrule
      \end{tabular}
    \end{center}
  \end{minipage}
  \label{tab:experimental-settings}
\end{table}

Because $x_1, \ldots, x_n$ define an equispaced grid of size $n$ on
$[-1, 1]$, the parameters $\beta_1, \beta_2, \gamma_1, \gamma_2$ are
uniquely determined by $\mu_1$, $\mu_n$, $\phi_1$, and $\phi_n$. Also,
for $\mu_1 \le \mu_n$ and $\phi_1 \le \phi_n$, it necessarily holds
that $\mu_i \in [\mu_1, \mu_n]$ and $\phi_i \in [\phi_1,
\phi_n]$. Table~\ref{tab:experimental-settings} shows the intervals
$[\mu_1, \mu_n]$ and $[\phi_1, \phi_n]$ we consider in the simulation
experiment, along with the implied parameter values. The intervals
cover narrow and wide ranges for $\mu$, being either close to the
boundary values or around 0.5, and narrow and wide intervals for
$\phi_i$, with both low and high values of the precision parameters.

We set $n = 500$. In order to avoid trivial response configurations,
from the possible combinations of
$u \in \{2^{-6}, 2^{-5}, \ldots, 2\}$, and $[\mu_1, \mu_n]$ and
$[\phi_1, \phi_n]$ in Table~\ref{tab:experimental-settings}, we
consider those for which the average probability of non-boundary
response values in $(0, 1)$ is more than $0.05$.  We then simulate
$100$ independent vectors $(y_1, \ldots, y_n, y_1', \ldots, y_n')$ per
parameter setting, and compute the relative change in $S$,
$S_{\rm CN} / S_{\rm XBX} - 1$, when moving from the XBX to the CN
model, based on their maximum likelihood fits. Positive values of the
relative change correspond to $S_{\rm CN} > S_{\rm XBX}$, and, thus,
worse predictive performance of the CN model relative to the XBX
regression model.

Figure~\ref{fig:xbx-vs-htobit} shows the average relative change for a
subset of the parameter settings we consider; see,
Figure~\ref{fig-supp:xbx-vs-htobit} in the supplementary material
document, for all settings. The XBX model is performing the same or
outperforms the CN model in the majority of settings. CN is found to
perform slightly better when the response distributions have small
variance (large values of $\phi$) for moderate values of $u$. Those
settings produce datasets where the relationship between the
non-boundary responses and covariates is well-modelled as linear; see
Figure~\ref{fig-supp:xb-data-1} to Figure~\ref{fig-supp:xb-data-4} in
the supplementary material document, which show simulated data sets
for all parameter settings. The two models start performing similarly
again for large $u$, when the response values start concentrating on
the boundaries.

\begin{figure}[t!]
  \begin{center}
    \includegraphics[width = \textwidth]{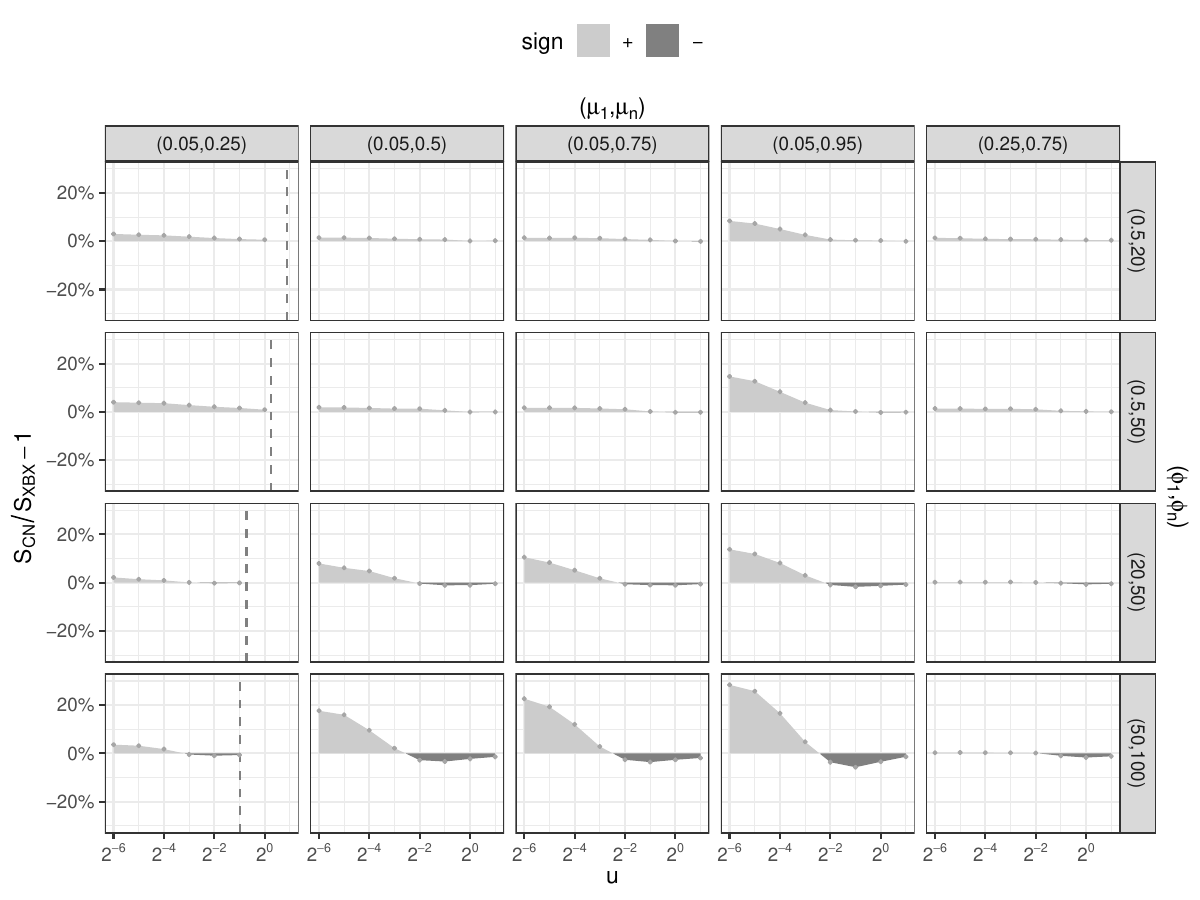}
  \end{center}
  \caption{Average relative change in $S$ in~(\ref{eq:sum-crps})
    ($S_{\rm CN} / S_{\rm XBX} - 1$) when moving from XBX to CN, based
    on $100$~samples of size $n = 500$ for each combination of
    $u \in \{2^{-6}, 2^{-5}, \ldots, 2\}$, and a subset of
    $[\mu_1, \mu_n]$ (columns) and $[\phi_1, \phi_n]$ (rows) intervals
    in Table~\ref{tab:experimental-settings}. Positive values of the
    relative change correspond to $S_{\rm CN} > S_{\rm XBX}$, and,
    thus, worse predictive performance of the CN model relative to the
    XBX regression model. The dashed vertical lines are the smallest
    values of $u$ for which the average probability of boundary
    observations exceeds $0.95$ for each combination of
    $[\mu_1, \mu_n]$ and $[\phi_1, \phi_n]$.}
    \label{fig:xbx-vs-htobit}
  \end{figure}

  The experimental setting in the current section has been used to
  compare beta regression on rescaled responses according to the
  proposal in \citet{smithson+verkuilen:2006} and the normal linear
  regression model. Figure~\ref{fig-supp:xbx-vs-beta} and
  Figure~\ref{fig-supp:xbx-vs-ols} in the
  Section~\ref{sec:XBX-vs-beta-ols} of the supplementary material
  document show the average relative change in $S$ when moving from
  XBX regression to each of those models as a function of $u$ for all
  parameter settings in Table~\ref{tab:experimental-settings}. In both
  cases XBX regression performs best and increasingly better as $u$
  increases.

\section{Concluding remarks}
\label{sec:conculding}

We have introduced XBX regression for modelling bounded responses with
or without boundary observations. At the core of the XBX regression
model is the XB distribution, which is a censored version of a
four-parameter beta distribution with the same exceedance $u$ on the
left and right of $(0, 1)$. We prove that both beta regression with
dispersion effects and the heteroscedastic two-limit tobit model are
special cases of XB regression, depending on whether $u$ is zero or
infinity, respectively. To avoid identifiability issues that may arise
in estimating $u$ due to the similarity of the beta and normal
distribution for specific parameter settings that result in low
skewness and kurtosis, we define XBX regression as a shrunken version
of XB regression towards beta regression, by letting $u$ have an
exponential distribution with mean $\nu$. The associated marginal
likelihood can be conveniently and accurately approximated using a
Gauss-Laguerre quadrature rule, resulting in efficient estimation and
inference procedures. In Section~\ref{sec:xbx-vs-htobit}, we produce
evidence that XBX regression can substantially outperform
heteroscedastic two-limit tobit regression when the common features of
skewness and kurtosis of bounded-domain variables are encountered, and
perform only slightly worse when skewness and kurtosis are small.

In Section~\ref{sec:lossaversion}, we use XBX regression to analyze
investment decisions in the behavioral economics experiment of
\cite{glaetzleruetzler+sutter+zeileis:2015}, where the occurrence and
extent of loss aversion are of interest. In contrast to standard
approaches, XBX regression can simultaneously capture the probability
of rational behavior and the mean amount of loss aversion.

The XBX regression model can be used to model bounded responses with
or without boundary observations in arbitrary intervals $[a, b]$, for
known $a$ and $b$ by simply transforming the observed responses $y_i'$
to $y_i = (y_i - a) / (b - a)$. Furthermore, the assumption of an
exponential distribution for $u$ has been made just for
convenience. The exponential distribution has only one parameter, and
its simple form allows the use of Gauss-Laguerre quadrature for
maximum approximate likelihood, also giving rise to simple marginal
moments for the latent uncensored variable;
see~(\ref{eq:b4mix-mean-var}). Other assumptions for the distribution
of $u$ can also result in shrinkage towards beta regression. We
believe that the benefits in applications will, at best, be marginal,
at the expense of potentially complicating estimation and key
characteristics of the latent distribution. It is also tempting to
allow for regression structures in $\nu$. However,
by~(\ref{eq:b4mix-mean-var}), $\nu$ controls the rate that the
marginal variance of the latent distribution changes with the
precision parameter~$\phi$. As a result, regression structures on both
$\nu$ and $\phi$ can offer little extra flexibility compared to a
regression structure on $\phi$ and estimation of $\nu$. Additional
flexibility can be achieved in a more interpretable manner by
exploring finite mixtures of XBX regression models, similarly to how
\citet{gruen+kosmidis+zeileis:2012} defined finite mixtures of beta
regression models.

The \pkg{betareg} R package, since version 3.2-0, provides methods for
fitting and drawing inference from XBX regression models. Bayesian
inference for XBX regression is straightforward through template
modelling frameworks such as JAGS \citep{plummer:2003} and Stan
\citep{stan2017}, due to the alternative definition of XBX regression
in Section~\ref{sec:xbx} as the censored version of a mixture of
four-parameter beta regressions with the same
exponentially-distributed exceedance on the left and right of
$(0, 1)$. Methods based on Stan are provided by the \pkg{brms} R
\citep{brms} package since version 2.23.

\section*{Supplementary material}
The supplementary material is available at
\url{https://github.com/ikosmidis/XBX-supplementary} and provides the
supplementary material document referenced above, and instructions and
scripts to reproduce all numerical results and figures in the main
text and the supplementary material document.

\section*{Declarations}

For the purpose of open access, the authors have applied a Creative
Commons Attribution (CC BY) license to any Author Accepted Manuscript
version arising from this submission.

\bibliography{xbx}

\begin{thebibliography}{38}
\newcommand{\bibenquote}[1]{``#1''}
\providecommand{\natexlab}[1]{#1}
\providecommand{\url}[1]{\texttt{#1}}
\providecommand{\urlprefix}{URL }
\expandafter\ifx\csname urlstyle\endcsname\relax
  \providecommand{\doi}[1]{doi:\discretionary{}{}{}#1}\else
  \providecommand{\doi}{doi:\discretionary{}{}{}\begingroup
  \urlstyle{rm}\Url}\fi
\providecommand{\eprint}[2][]{\url{#2}}

\bibitem[{Abramowitz and Stegun(1964)}]{abramowitz+stegun:1964}
Abramowitz M, Stegun IA (1964).
\newblock \emph{Handbook of Mathematical Functions with Formulas, Graphs, and
  Mathematical Tables}.
\newblock Dover, New York, 9th edition.
\newblock ISBN 0-486-61272-4.

\bibitem[{Allenbrand and Sherwood(2023)}]{allenbrand+sherwood:2023}
Allenbrand C, Sherwood B (2023).
\newblock \bibenquote{Model Selection Uncertainty and Stability in Beta
  Regression Models: {A} Study of Bootstrap-Based Model Averaging with an
  Empirical Application to Clickstream Data.}
\newblock \emph{The Annals of Applied Statistics}, \textbf{17}(1), 680--710.
\newblock \doi{10.1214/22-AOAS1647}.

\bibitem[{Benartzi and Thaler(1995)}]{benartzi+thaler:1995}
Benartzi S, Thaler RH (1995).
\newblock \bibenquote{Myopic Loss Aversion and the Equity Premium Puzzle.}
\newblock \emph{The Quarterly Journal of Economics}, \textbf{110}(1), 73--92.
\newblock \doi{10.2307/2118511}.

\bibitem[{B\"urkner(2017)}]{brms}
B\"urkner PC (2017).
\newblock \bibenquote{{brms}: An {R} Package for {Bayesian} Multilevel Models
  Using {Stan}.}
\newblock \emph{Journal of Statistical Software}, \textbf{80}(1), 1--28.
\newblock \doi{10.18637/jss.v080.i01}.

\bibitem[{Calabrese(2012)}]{calabrese:2012}
Calabrese R (2012).
\newblock \bibenquote{Regression Model for Proportions with Probability Masses
  at Zero and One.}
\newblock \emph{Working Papers 201209}, Geary Institute, University College
  Dublin.
\newblock \urlprefix\url{https://EconPapers.repec.org/RePEc:ucd:wpaper:201209}.

\bibitem[{Carpenter \emph{et~al.}(2017)Carpenter, Gelman, Hoffman, Lee,
  Goodrich, Betancourt, Brubaker, Guo, Li, and Riddell}]{stan2017}
Carpenter B, Gelman A, Hoffman MD, Lee D, Goodrich B, Betancourt M, Brubaker M,
  Guo J, Li P, Riddell A (2017).
\newblock \bibenquote{{S}tan: {A} Probabilistic Programming Language.}
\newblock \emph{Journal of Statistical Software}, \textbf{76}(1), 1--32.
\newblock \doi{10.18637/jss.v076.i01}.

\bibitem[{Cook \emph{et~al.}(2008)Cook, Kieschnick, and
  McCullough}]{cook+kieschnick+mccullough:2008}
Cook DO, Kieschnick R, McCullough BD (2008).
\newblock \bibenquote{Regression Analysis of Proportions in Finance with Self
  Selection.}
\newblock \emph{Journal of Empirical Finance}, \textbf{15}(5), 860--867.
\newblock ISSN 0927-5398.
\newblock \doi{10.1016/j.jempfin.2008.02.001}.

\bibitem[{Cribari-Neto and Zeileis(2010)}]{betareg}
Cribari-Neto F, Zeileis A (2010).
\newblock \bibenquote{Beta Regression in {R}.}
\newblock \emph{Journal of Statistical Software}, \textbf{34}(2), 1--24.
\newblock \doi{10.18637/jss.v034.i02}.

\bibitem[{Dixon and Sonka(1982)}]{dixon+sonka:1982}
Dixon BL, Sonka ST (1982).
\newblock \bibenquote{Use of Two-Limit Probit Regression Model: An Analysis of
  Lender Response to Loan Requests.}
\newblock \emph{Journal of Business Research}, \textbf{10}(4), 489--502.
\newblock ISSN 0148-2963.
\newblock \doi{10.1016/0148-2963(82)90007-8}.

\bibitem[{Douma and Weedon(2019)}]{douma+weedon:2019}
Douma JC, Weedon JT (2019).
\newblock \bibenquote{Analysing Continuous Proportions in Ecology and
  Evolution: {A} Practical Introduction to Beta and {D}irichlet Regression.}
\newblock \emph{Methods in Ecology and Evolution}, \textbf{10}(9), 1412--1430.
\newblock \doi{10.1111/2041-210X.13234}.

\bibitem[{Ferrari and Cribari-Neto(2004)}]{ferrari+cribari-neto:2004}
Ferrari SLP, Cribari-Neto F (2004).
\newblock \bibenquote{Beta Regression for Modelling Rates and Proportions.}
\newblock \emph{Journal of Applied Statistics}, \textbf{31}(7), 799--815.
\newblock \doi{10.1080/0266476042000214501}.

\bibitem[{Geissinger \emph{et~al.}(2022)Geissinger, Khoo, Richmond, Faulkner,
  and Schneider}]{geissinger+khoo+richmond:2022}
Geissinger EA, Khoo CLL, Richmond IC, Faulkner SJM, Schneider DC (2022).
\newblock \bibenquote{A Case for Beta Regression in the Natural Sciences.}
\newblock \emph{Ecosphere}, \textbf{13}(2), e3940.
\newblock \doi{10.1002/ecs2.3940}.

\bibitem[{Gl\"atzle-R\"utzler \emph{et~al.}(2015)Gl\"atzle-R\"utzler, Sutter,
  and Zeileis}]{glaetzleruetzler+sutter+zeileis:2015}
Gl\"atzle-R\"utzler D, Sutter M, Zeileis A (2015).
\newblock \bibenquote{No Myopic Loss Aversion in Adolescents? An Experimental
  Note.}
\newblock \emph{Journal of Economic Behavior \& Organization}, \textbf{111},
  169--176.
\newblock \doi{10.1016/j.jebo.2014.12.021}.

\bibitem[{Gneiting and Raftery(2007)}]{gneiting+raftery:2007}
Gneiting T, Raftery AE (2007).
\newblock \bibenquote{Strictly Proper Scoring Rules, Prediction, and
  Estimation.}
\newblock \emph{Journal of the American Statistical Association},
  \textbf{102}(477), 359--378.
\newblock \doi{10.1198/016214506000001437}.

\bibitem[{Greene(2011)}]{greene:2011}
Greene WH (2011).
\newblock \emph{Econometric Analysis}.
\newblock Prentice Hall, Upper Saddle River, NJ, 7th edition.

\bibitem[{Gr\"un \emph{et~al.}(2012)Gr\"un, Kosmidis, and
  Zeileis}]{gruen+kosmidis+zeileis:2012}
Gr\"un B, Kosmidis I, Zeileis A (2012).
\newblock \bibenquote{Extended Beta Regression in {R}: Shaken, Stirred, Mixed,
  and Partitioned.}
\newblock \emph{Journal of Statistical Software}, \textbf{48}(1), 1--25.
\newblock ISSN 1548-7660.
\newblock \doi{10.18637/jss.v048.i11}.

\bibitem[{Hoff(2007)}]{hoff:2007}
Hoff A (2007).
\newblock \bibenquote{Second Stage {DEA}: Comparison of Approaches for
  Modelling the {DEA} Score.}
\newblock \emph{European Journal of Operational Research}, \textbf{181}(1),
  425--435.
\newblock ISSN 0377-2217.
\newblock \doi{10.1016/j.ejor.2006.05.019}.

\bibitem[{Johnson \emph{et~al.}(1995)Johnson, Kotz, and
  Balakrishnan}]{johnson+kotz+balakrishnan:1995}
Johnson NL, Kotz S, Balakrishnan N (1995).
\newblock \emph{Continuous Univariate Distributions}.
\newblock Wiley Series in Probability and Mathematical Statistics: Applied
  Probability and Statistics. John Wiley \& Sons, 2 edition.
\newblock ISBN 9780471584940.

\bibitem[{Kieschnick and McCullough(2003)}]{kieschnick+mccullough:2003}
Kieschnick R, McCullough BD (2003).
\newblock \bibenquote{Regression Analysis of Variates Observed on $(0, 1)$:
  Percentages, Proportions and Fractions.}
\newblock \emph{Statistical Modelling}, \textbf{3}(3), 193--213.
\newblock \doi{10.1191/1471082X03st053oa}.

\bibitem[{Kleiber and Zeileis(2016)}]{kleiber+zeileis:2016}
Kleiber C, Zeileis A (2016).
\newblock \bibenquote{Visualizing Count Data Regressions Using Rootograms.}
\newblock \emph{The American Statistician}, \textbf{70}(3), 296--303.
\newblock \doi{10.1080/00031305.2016.1173590}.

\bibitem[{Leng \emph{et~al.}(2021)Leng, Li, Eser, Piergies, Sit, Tan, Neff, Li,
  Rodriguez, Suemoto, Leite, Ehrenberg, Pasqualucci, Seeley, Spina, Heinsen,
  Grinberg, and Kampmann}]{leng+li+eser:2021}
Leng K, Li E, Eser R, Piergies A, Sit R, Tan M, Neff N, Li SH, Rodriguez RD,
  Suemoto CK, Leite REP, Ehrenberg AJ, Pasqualucci CA, Seeley WW, Spina S,
  Heinsen H, Grinberg LT, Kampmann M (2021).
\newblock \bibenquote{Molecular Characterization of Selectively Vulnerable
  Neurons in {A}lzheimer's Disease.}
\newblock \emph{Nature Neuroscience}, \textbf{24}(2), 276--287.
\newblock \doi{10.1038/s41593-020-00764-7}.

\bibitem[{Liu and Kong(2015)}]{liu+kong:2015}
Liu F, Kong Y (2015).
\newblock \bibenquote{{zoib}: An {R} Package for {B}ayesian Inference for Beta
  Regression and Zero/One Inflated Beta Regression.}
\newblock \emph{The R Journal}, \textbf{7}(2), 34--51.
\newblock \doi{10.32614/RJ-2015-019}.

\bibitem[{Maddala(1983)}]{maddala:1983}
Maddala GS (1983).
\newblock \emph{Limited-Dependent and Qualitative Variables in Econometrics}.
\newblock Econometric Society Monographs. Cambridge University Press.
\newblock ISBN 9780511810176.
\newblock \doi{10.1017/CBO9780511810176}.

\bibitem[{Messner \emph{et~al.}(2016)Messner, Mayr, and
  Zeileis}]{messner+mayr+zeileis:2016}
Messner JW, Mayr GJ, Zeileis A (2016).
\newblock \bibenquote{Heteroscedastic Censored and Truncated Regression with
  {crch}.}
\newblock \emph{The R Journal}, \textbf{8}(1), 173--181.
\newblock \doi{10.32614/rj-2016-012}.

\bibitem[{Min and Agresti(2002)}]{min+agresti:2002}
Min Y, Agresti A (2002).
\newblock \bibenquote{Modeling Nonnegative Data with Clumping at Zero: A
  Survey.}
\newblock \emph{Journal of the Iranian Statistical Society}, \textbf{1}, 7--33.
\newblock \urlprefix\url{https://jirss.irstat.ir/article_253597.html}.

\bibitem[{Moscovich \emph{et~al.}(2016)Moscovich, Nadler, and
  Spiegelman}]{moscovich+etal:2016}
Moscovich A, Nadler B, Spiegelman C (2016).
\newblock \bibenquote{On the Exact {B}erk-{J}ones Statistics and Their
  $p$-Value Calculation.}
\newblock \emph{Electronic Journal of Statistics}, \textbf{10}(2), 2329--2354.
\newblock \doi{10.1214/16-EJS1172}.

\bibitem[{Namin \emph{et~al.}(2020)Namin, Xu, Zhou, and
  Beyer}]{namin+xu+zhou:2020}
Namin S, Xu W, Zhou Y, Beyer K (2020).
\newblock \bibenquote{The Legacy of the {H}ome {O}wners' {L}oan {C}orporation
  and the Political Ecology of Urban Trees and Air Pollution in the {U}nited
  {S}tates.}
\newblock \emph{Social Science \& Medicine}, \textbf{246}, 112758.
\newblock \doi{10.1016/j.socscimed.2019.112758}.

\bibitem[{Ospina and Ferrari(2010)}]{ospina+ferrari:2010}
Ospina R, Ferrari SLP (2010).
\newblock \bibenquote{Inflated Beta Distributions.}
\newblock \emph{Statistical Papers}, \textbf{51}(1), 111--126.
\newblock \doi{10.1007/s00362-008-0125-4}.

\bibitem[{Ospina and Ferrari(2012)}]{ospina+ferrari:2012}
Ospina R, Ferrari SLP (2012).
\newblock \bibenquote{A General Class of Zero-or-One Inflated Beta Regression
  Models.}
\newblock \emph{Computational Statistics \& Data Analysis}, \textbf{56}(6),
  1609--1623.
\newblock \doi{10.1016/j.csda.2011.10.005}.

\bibitem[{Papke and Wooldridge(1996)}]{papke+wooldridge:1996}
Papke LE, Wooldridge JM (1996).
\newblock \bibenquote{Econometric Methods for Fractional Response Variables
  with an Application to {401(k)} Plan Participation Rates.}
\newblock \emph{Journal of Applied Econometrics}, \textbf{11}(6), 619--632.
\newblock ISSN 08837252, 10991255.
\newblock \doi{10.1002/(sici)1099-1255(199611)11:6<619::aid-jae418>3.0.co;2-1}.

\bibitem[{Plummer(2003)}]{plummer:2003}
Plummer M (2003).
\newblock \bibenquote{{JAGS}: {A} Program for Analysis of {B}ayesian Graphical
  Models Using {G}ibbs Sampling.}
\newblock In K~Hornik, F~Leisch, A~Zeileis (eds.), \bibenquote{Proceedings of
  the 3rd International Workshop on Distributed Statistical Computing (DSC
  2003),} Technische Universit{\"a}t Wien, Vienna, Austria.
\newblock
  \urlprefix\url{https://www.R-project.org/conferences/DSC-2003/Proceedings/Plummer.pdf}.

\bibitem[{Rigby and Stasinopoulos(2005)}]{rigby+stasinopoulos:2005}
Rigby RA, Stasinopoulos DM (2005).
\newblock \bibenquote{Generalized additive models for location, scale and shape
  (with discussion).}
\newblock \emph{Applied Statistics}, \textbf{54}, 507--554.
\newblock \doi{10.1111/j.1467-9876.2005.00510.x}.

\bibitem[{Rosett and Nelson(1975)}]{rosett+nelson:1975}
Rosett RN, Nelson FD (1975).
\newblock \bibenquote{Estimation of the Two-Limit Probit Regression Model.}
\newblock \emph{Econometrica}, \textbf{43}(1), 141--146.
\newblock \doi{10.2307/1913419}.

\bibitem[{Smithson and Merkle(2013)}]{smithson+merkle:2013}
Smithson M, Merkle EC (2013).
\newblock \emph{Generalized Linear Models for Categorical and Continuous
  Limited Dependent Variables}.
\newblock Chapman \& Hall/CRC, New York.
\newblock ISBN 9780429185311.
\newblock \doi{10.1201/b15694}.

\bibitem[{Smithson and Verkuilen(2006)}]{smithson+verkuilen:2006}
Smithson M, Verkuilen J (2006).
\newblock \bibenquote{A Better Lemon Squeezer? {M}aximum-Likelihood Regression
  with Beta-Distributed Dependent Variables.}
\newblock \emph{Psychological Methods}, \textbf{11}(1), 54--71.
\newblock \doi{10.1037/1082-989X.11.1.54}.

\bibitem[{Smyth(2005)}]{smyth:2005}
Smyth GK (2005).
\newblock \bibenquote{Numerical Integration.}
\newblock In P~Armitage, T~Colton (eds.), \bibenquote{Encyclopedia of
  Biostatistics,} pp. 3088--3095. John Wiley \& Sons.
\newblock \doi{10.1002/0470011815.b2a14026}.

\bibitem[{Wang \emph{et~al.}(2015)Wang, Malthouse, and
  Krishnamurthi}]{wang+malthouse+krishnamurthi:2015}
Wang RJH, Malthouse EC, Krishnamurthi L (2015).
\newblock \bibenquote{On the Go: {H}ow Mobile Shopping Affects Customer
  Purchase Behavior.}
\newblock \emph{Journal of Retailing}, \textbf{91}(2), 217--234.
\newblock \doi{10.1016/j.jretai.2015.01.002}.

\bibitem[{Winkelmann and Boes(2006)}]{winkelmann+boes:2006}
Winkelmann R, Boes S (2006).
\newblock \emph{Analysis of Microdata}.
\newblock Springer, New York, NY, 1st edition.
\newblock ISBN 978-3-540-29605-8.

\end{thebibliography}
\bibliographystyle{jss2}

\includepdf[pages=-]{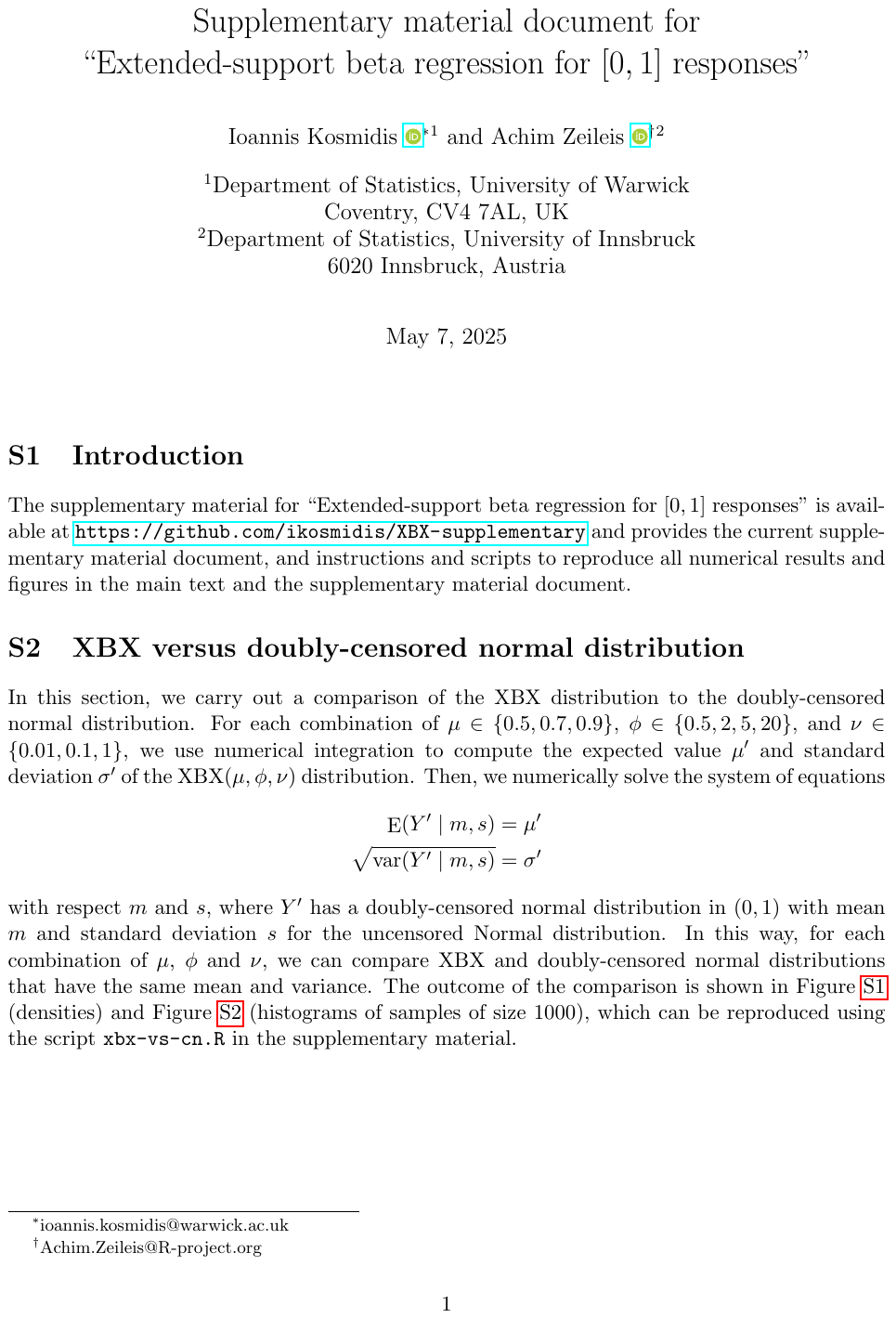}

\end{document}